\documentclass[10pt,a4paper]{article}

\usepackage[caption=false]{subfig}  

\usepackage{amsmath,amssymb}

\usepackage{changepage}

\usepackage{textcomp,marvosym}

\usepackage{nameref,hyperref}

\usepackage{microtype}
\DisableLigatures[f]{encoding = *, family = * }

\usepackage[table]{xcolor}

\usepackage{array}

\newcolumntype{+}{!{\vrule width 2pt}}

\newlength\savedwidth

\usepackage[aboveskip=1pt,labelfont=bf,labelsep=period,justification=raggedright,singlelinecheck=off]{caption}

\usepackage{xcolor}

\usepackage{tikz}

\definecolor{mygreen}{rgb}{0,0.6,0}

\newcommand{\tmin}{t_{\min}}
\newcommand{\tmax}{t_{\max}}

\usepackage{amsmath}
\usepackage{amssymb}
\usepackage{amsthm}
  \newtheorem{theorem}{Theorem}
  \newtheorem{lemma}[theorem]{Lemma} 
  \newtheorem{definition}[theorem]{Definition}

\graphicspath{{plots}{images}}

\title{The Curse of Shared Knowledge: Recursive Belief Reasoning in a Coordination Game with Imperfect Information}

\author{Thomas Bolander\textsuperscript{1\Yinyang}, Robin Engelhardt\textsuperscript{2\Yinyang}, Thomas S. Nicolet\textsuperscript{3\Yinyang}} 

\begin{document}

\begingroup
\renewcommand\thefootnote{}\footnotetext{
   \textbf{1} Department of Applied Mathematics and Computer Science, Technical University of Denmark, Richard Petersens Plads, building 324, DK-2800 Lyngby, Denmark

\textbf{2} 
Department of Food and Resource Economics, University of Copenhagen, Copenhagen, Denmark
 
\textbf{3}
Center for Information and Bubble Studies, Department of Communication, University of Copenhagen, Copenhagen, Denmark.

%
%
\Yinyang All authors contributed equally to this work.
}%
\addtocounter{footnote}{-1}%
\endgroup







\maketitle

\begin{abstract}
Common knowledge is crucial for safe group coordination. In its absence, 
humans must 
rely on shared knowledge, 
which is inherently limited in depth and therefore prone to coordination failures,
because any finite-order knowledge attribution allows for an even higher order attribution that may change what is known by whom. 
In three separate experiments involving 802 participants, we investigate the extent to which humans 
can differentiate between common knowledge and $n$th-order shared knowledge. We 
designed a two-person coordination game with imperfect information to simplify the 
recursive game structure  and higher-order uncertainties into a relatable everyday scenario. 
In this game, coordination 
for the highest payoff requires a specific fact to be common knowledge between players. 
However, this fact cannot become common knowledge 
in the game. The fact can at most be $n$th-order shared knowledge for some $n$.
Our findings reveal that even at quite shallow depths of shared knowledge (low values of $n$), players behave as though they possess common knowledge, and 
claim similar levels of certainty in their actions, 
despite incurring significant penalties when falsely assuming guaranteed coordination.
We term this phenomenon `the curse of shared knowledge'. It 
arises either from the players' 
inability to distinguish between higher-order shared knowledge and common knowledge, or from their implicit assumption that their co-player cannot make this distinction.

\end{abstract}

\section{Introduction}
Anybody who has been caught up in a sidewalk dance knows that human coordination without common knowledge about each other's intentions is prone to failure. Thus, cooperating humans have evolved techniques by which to create common knowledge in practice, such as eye contact, conventions, and broadcasted messages. When common knowledge cannot be obtained, humans routinely act upon their private knowledge, or upon their shared reasoning about other minds. While we know much about human limits to such higher-order belief reasoning, little is known about at which point, and under which circumstances, shared knowledge becomes effectively indistinguishable from common knowledge. We investigate this question with a novel recursive coordination game with imperfect information, and show that people consistently act as if they have common knowledge about a fact, even in situations where they share only the shallowest levels of their state of mind.  

Successful group coordination requires complementary choices among group members, which, in turn, requires communication of beliefs and intentions in such a way that they become common knowledge~\cite{fagin1995reasoning}. A fact is said to be \textit{common knowledge} if everyone knows it, and everyone knows everyone knows it, and everyone knows everyone knows everyone knows it, and so on, ad infinitum~\cite{lewis1969convention, clark1981definite, schelling1980strategy, aumann1976agreeing}. If the premise of  ``everyone knows'' is not infinitely nested, but only nested to finite depth, we instead have \emph{shared knowledge}. If there is no nested knowledge about knowledge at all and not everyone necessarily knows the fact, we have \emph{private knowledge}.

Let us illustrate the difference between these notions with an example. Two friends, Agnes and Bertram, are taking different trekking routes on the same mountain, and have agreed to meet at the top. 
The equipment essential for an overnight stay at the top has been divided between their backpacks, so it's crucial for their safety that they meet there before it gets dark. 
On the way to the top, they both observe bad weather approaching, which makes it safer to go back to the base. However, they are both uncertain about whether the other person has seen the bad weather, so they cannot expect the other person to necessarily go to the base. The most dangerous situation is if one of them goes to the base and the other goes to the top, since then one of them is left at the top without the necessary equipment to stay there for the night. 

At this point they both know the fact that ``bad weather approaches'', but don't know whether the other knows. In this situation, we would say that Agnes and Bertram both have \emph{private knowledge} that bad weather is approaching, and since they both know it, it is also \emph{shared knowledge} between them. More generally, a fact is \emph{private knowledge} in a group of agents if some non-empty subset of the agents know the fact. For the special case where everybody in the group knows it, we say that there is \emph{shared knowledge to depth one} (or \emph{first-order shared knowledge}) of the fact~\cite{clark1981definite}. The exact border between what is considered private and shared knowledge vary between authors. Some consider the case where everybody knows a fact to still only be private knowledge, and that for it to be shared knowledge, at least one agent also needs to know that they all know. The notion of shared knowledge we use in this paper, where shared knowledge only requires that everybody knows, is the one originally used by~\cite{clark1981definite}, and also the one often seen in (epistemic) logic~\cite{herzig2015share}. Formal definitions of private, shared and common knowledge as well as a more detailed discussion of our choice of conventions can be found in Appendix \ref{A:definitions}. 

Since Agnes doesn't know whether Bertram knows about the bad weather, she would like to warn him and suggest to meet back at the base. So she sends a text message: ``Bad weather approaching. Let's meet at base.'' However, due to the unstable mobile network signals, she is not certain that the message will go through. Therefore she asks Bertram to confirm that he received the  message. A few minutes later, she receives his confirmation. At this point it has become \emph{shared knowledge to depth two} (or \emph{second-order shared knowledge}) that bad weather is approaching: He knows that she knows, since he received her message, and she knows that he knows, because she received his confirmation. To have \emph{shared knowledge to depth three} (\emph{third-order shared knowledge}), it would additionally be required that A) she knows that he knows that she knows, and B) that he knows that she knows that he knows. In fact, A already holds, since she received (and read) his confirmation. However, B doesn't hold, since he will be uncertain about whether his confirmation was successfully received. Thus we have asymmetry in the level of knowledge of the two agents. If he also asks her to confirm his message, and he receives such a confirmation, then of course he will get to know that she knows that he knows. Then there will be shared knowledge to depth three. However, there will still be a (higher-order) knowledge asymmetry, since Agnes can't be certain that the last message was received.

How many messages back and forth does it take for Agnes and Bertram to coordinate going back to the base? At first it might seem that it is sufficient for both of them to know that at least one of them plans to go to the base. However, that is not so. After the first message has been received, both know that Agnes plans to go to the base, but she is still uncertain whether he knows. And if he doesn't, she might risk leaving him alone at the top. So shared knowledge to depth one is clearly not sufficient. Even shared knowledge to depth two is not sufficient. After he has confirmed receiving her original message, from his perspective it is still entirely possible that she doesn't know that he knows, and hence, she might decide to go to the top to not leave him alone up there. And in this case, he also has to go to the top in order not to leave \emph{her} alone up there. This argument can be generalised to prove that even $n$th-order shared knowledge for any arbitrary large $n$ is insufficient for safe coordination.\footnote{To derive a contradiction, suppose $n$th-order shared knowledge for some $n$ is sufficient to make it safe to go to the base, in the sense of guaranteeing that the other person will also go there. This implies that the successful delivery of the $n$th message is sufficient to guarantee meeting at the base. There must then exist a smallest number $n_0$ such that the successful delivery of the $n_0$th message is sufficient to guarantee meeting at the base. Since $n_0$ is the smallest such number, the successful delivery of the $(n_0-1)$st message is not sufficient to guarantee meeting at the base. Now note that at the moment when the $n_0$th message has been successfully delivered, the sender of this message is still uncertain about whether it was actually received, and hence the sender is only certain that the first $n_0-1$ messages has been successfully delivered.  In other words, the sender of the $n_0$th message considers it possible---even after the successful delivery of the message---that only the first $n_0-1$ messages were successfully delivered, in which case it is \emph{not} safe to go to the base. Hence that person will after having sent the $n_0$th message stil consider it unsafe to go to the base, and will choose to go to the top. This is a contradiction, completing the proof. This proof is stated in rather informal terms, but can be turned into a formal, mathematical proof~\cite{fagin1995reasoning}. The example here is in fact modelled after the coordinated attack problem studied extensively in that book, where it is also used to illustrate that common knowledge is a necessary condition for coordination. We are making a similar point here about the necessity of common knowledge for coordination, just using another example.}  The consequence is that no finite number of messages successfully delivered will guarantee that Agnes and Bertram manage to meet at the base. In order to guarantee meeting at the base, they would need to have shared $n$th-order knowledge for \emph{all} $n$. Shared $n$th-order knowledge of some fact for all $n$ is called  \emph{common knowledge} of the fact (see again  
Appendix \ref{A:definitions} for a more formal definition). The essence of the argument is that any  
plan for Agnes and Bertram to coordinate without common knowledge would leave one stranded at the top in the case of a faulty signal at a critical point in their communication.

In practice, shouldn't it be possible to coordinate meeting at the base after one or two messages being sent back and forth? It is likely that this is what Agnes and Bertram would attempt, at least. The question is whether they are fully aware of the risk of miscoordination in this case? After Agnes has sent the message and Bertram has successfully confirmed receiving it, are they aware that the content of the message is still not common knowledge, and that they still risk miscoordination? Or do they now believe it completely safe to go to the base? The questions here are whether A) they can distinguish between the case of $n$th-order shared knowledge and common knowledge, and whether B) they would act differently in the case of common knowledge than in the case of $n$th-order shared knowledge for some suitable large $n$. 

We have designed a coordination game that has a similar underlying mathematical structure as this trekking example, but cast in a simpler setting, where the higher-order uncertainty is established already when the game is initialised and not, as above, through a series of message passings. In this game, our experiments show that human players indeed try to coordinate on the preferred outcome even when only having low orders of shared knowledge (corresponding to trying to meet at the base after only a  few message passings in the trekking example). However, this often results in severe cases of miscoordination with high penalty, corresponding to the situation in the trekking example where one of the trekkers would leave the other at the top without the necessary equipment. Our game suggests that the answer to B above is sometimes a ``no'': Most players definitely act as if they have common knowledge, even though they don't, i.e., for sufficiently high levels of shared knowledge, they act as in the common knowledge situation. For sufficiently high levels of shared knowledge, they play strategies that would have given them the highest payoff if they had had common knowledge, but since they don't, they end up with payoffs that are significantly lower than what they could have received if they had explicitly distinguished the common knowledge situation from the $n$th-order shared knowledge situation. One might however still speculate that this is a 
conscious choice of strategy, and that the players are nevertheless aware that they don't have common knowledge. However, through follow-up questions and interviews, it seems that the answer to A is also ``no'': Players claim they have common knowledge and that coordination is guaranteed, even when they only have low-order shared knowledge.     

Humans generally find the concept of common knowledge hard to grasp, and even harder to grasp the practical relevance of, due to the unbounded nesting of knowledge. We originally developed our game to illustrate that the difference between $n$th-order shared knowledge and common knowledge can have actual strategic implications, hence helping to understand the notion of common knowledge better (and why $n$th-order shared knowledge for sufficiently high values of $n$ is not sufficient). The game we developed makes it clear that human intuitions about common knowledge can be misleading and may have costly consequences.

\section{The curse of shared knowledge}
Reasoning about the knowledge of others, their reasoning about you, and your reasoning about their reasoning, and so on, is famous in cognitive science for its presumed computational intractability~\cite{van2018parameterized,bolander2023parameterized}. Because of this, coordinating species typically use heuristic shortcuts in order to reduce complex shared knowledge states, such as joint perceptual cues and broadcasted signals~\cite{milgrom1981axiomatic, clark1996using, bradbury1998principles}. Humans may obtain common knowledge via mutually accessible first-order sensory experiences~\cite{tomasello1995joint, lorini2005establishing, bolander2015announcements, gintis2010rationality}, eye contact~\cite{friedell1969structure}, public rituals and conventions~\cite{lewis1969convention}, or salient focal points~\cite{schelling1957bargaining}. What most prominently is believed to distinguish human coordination from other animals, however, is the enormous flexibility by which humans can imagine and articulate the mental states of their peers~\cite{tooby2010groups, harari2014sapiens}. The abilities to blush, to tell jokes, and to write novels testify that humans readily attribute \textit{higher-order} beliefs, intentions, and reasoning capabilities to other people, such as thinking explicitly about the mental states of others who think about the thoughts and beliefs of others and so on---while at the same time appreciating that those thoughts and beliefs can differ from each other and from reality. Such higher-order cognition has seen substantial scientific attention, and has brought about various technical terms such as ``Theory of Mind'' (ToM)~\cite{premack1978does}, ``mentalizing''~\cite{frith2003development}, ``mind reading''~\cite{vogeley2001mind, apperly2010mindreaders}, 
``mind perception''~\cite{gray2011distortions}, ``perspective taking'', and ``social intelligence''~\cite{baron1999social}, which often are used interchangeably for studying the cognitive mechanisms of shared knowledge, but sometimes focus on slightly different ideas and associated meanings~\cite{schaafsma2015deconstructing}, depending on the field of investigation.

Looking at the ToM-literature, conclusions about human belief reasoning abilities are rather heterogeneous~\cite{apperly2009humans, saxe2013theory}: Human reasoning about the reasoning processes of other humans is limited~\cite{premack1978does, gopnik1988children, stahl1995players, nagel1995unraveling, hedden2002you, keysar2003limits, pinker2003language}, contextual, and possibly domain specific~\cite{leslie1992domain, saxe2009theory, heyes2014submentalizing}. Three-year-old children tend to fail in the well-known first-order false belief tasks by falsely assuming that their private information is shared by others, while second-order false belief tasks are mastered around ages 5-7~\cite{perner1985john, sullivan1994preschoolers}. Adults may reliably master up to four orders~\cite{kinderman1998theory}, but still have difficulties ignoring the private information they possess when assessing the beliefs of others, resulting in a \textit{curse of knowledge} bias which can compromise their ability to make predictions about other people’s beliefs and actions~\cite{camerer1989curse, birch2007curse}. When the nested mental states represent a succession of different people, such as ``Alice thinks that Bob thinks that Carol is contemplating the idea that David is thinking about Evelyn'', we have less problems following along\footnote{Especially when those successions of mental states are qualified by psychological attribute words such as `Alice thinks that Bob is mistakenly worrying that Carol is offended by misunderstanding something Dave had said to Evelyn'.~\cite{academian2019unrolling}.
} than when the nested mental states are successions of the same people over and over again, and thus are truly recursive, such as ``I think that you think that I contemplate the idea that you are thinking about me''. We get confused more easily by the latter formulation, since we need to keep track of several representations of ourselves and of the other, each representation differing in its perspective and in the number of mental states it presupposes~\cite{de2019common}. When humans compete or try to detect cheaters, higher-order belief reasoning seems to perform better than lower-order belief reasoning~\cite{goodie2012levels}. In negotiations and other mixed motive situations, where innuendo, threats, bribes and other kinds of indirect propositions are common, humans are very good at the \textit{strategic} use of higher-order belief reasoning, for instance as a means to \emph{prevent} common knowledge in certain groups of agents, or as a means to form specific knowledge alliances~\cite{pinker2007stuff, pinker2008logic, de2017negotiating}. ToM proficiency may also be facilitated by providing games with stepwise increase in ToM~\cite{verbrugge2018stepwise}.

In pure coordination problems, such as pedestrians choosing sides, or people agreeing on new words or on new technical standards, common knowledge is the normatively preferred informational state for all members of the group, because it allows to coordinate on an optimal common equilibrium. If there are no or limited means by which to communicate, however, people face an equilibrium-selection problem for which neither game theory nor the ToM literature has any clear solution. Although some experimental evidence~\cite{curry2012putting} suggests that higher-order ToM reasoning may improve coordination efforts, other work seems to suggest that coordination favours lower orders of ToM sophistication~\cite{ devaine2014theory, de2015higher}. The challenge of tacit coordination is particularly relevant for artificial intelligence research and for social cognitive robotics, where the implementation of ToM-like processes into artificial social agents is believed to be an important step towards reliable human-robot interaction~\cite{erb2016artificial, bolander2018seeing, bard2020hanabi, dissing2020implementing}.

Recently, researchers have investigated whether humans have adapted specifically to recognizing common knowledge as a separate cognitive category, distinct from both private and shared knowledge~\cite{de2019common}. Controlled pure coordination experiments in social settings on market collaboration~\cite{thomas2014psychology}, the bystander effect~\cite{thomas2016recursive}, indirect speech~\cite{lee2010rationales}, self-conscious emotions~\cite{thomas2018common}, and charity~\cite{de2019maimonides}, consistently find that people indeed make strategically different choices under common knowledge conditions (presented in the form of public announcements), compared to situations in which there is only private knowledge (in the form of private messages) or shared knowledge (private messages that elaborate on the depth of knowledge of other participants). Apart from seeing a clear benefit of common knowledge, some of these studies also showed that people have a hard time discriminating between various orders of shared knowledge, and that coordination efforts do not correlate with payoff conditions~\cite{thomas2014psychology}, which is in contrast to the assumptions of standard rational choice theory in which payoffs are expected to be maximized~\cite{becker1976economic}.

Most of these studies are designed one way or the other to determine how far people typically unpack shared knowledge and whether they can recognize common knowledge when presented with it. They show that humans are indeed able to discriminate common knowledge from shared knowledge, and adapt their actions accordingly, at least in cases where it is simple to discriminate between the two states of knowledge. The purpose of our paper is somewhat different. The purpose is not to determine whether players can for instance distinguish $n$th order and $(n+1)$th order shared knowledge, but whether they can in general distinguish shared knowledge (of some order) from common knowledge. In the above-mentioned experiments, common knowledge was achieved via public announcements, and shared knowledge via private announcements in which the depth of knowledge was explicated. In other words, these were experiments in which the distinction between shared and common knowledge was made clear and explicit. The question then remains how good humans are at recognizing the difference between shared and common knowledge when the difference is not made explicit, but has to be deduced? The fact that humans are able to adapt their actions to whether they are in a shared or common knowledge situation might have significantly less practical relevance if humans in general have a hard time to distinguish these in the wild. The results of this paper suggest that humans are indeed not very good at recognizing the difference, in particular that they are prone to mistake $n$th-order shared knowledge for common knowledge---even for relatively low values of $n$. Or, at least, they behave as if they had common knowledge even when they only have $n$th-order shared knowledge for very moderate values of $n$. 
This issue has not been thoroughly investigated in previous research, in part because many existing experimental designs stop after describing 2-3 orders of belief reasoning to the participants, as higher orders require quite convoluted sentences that tend to become incomprehensible and increase experimental error. Or, as in the mountain trekking example, they require reasoning about the consequences of a high number of (message passing) actions that each change the mental state of the involved agents.

The latter has been explored theoretically in the  `electronic mail game' by Rubinstein~\cite{rubinstein1989electronic}, a game version of the mountain trekking example and of the structurally equivalent `coordinated attack problem'~\cite{fagin1995reasoning}. The Rubinstein paper shows that `almost common knowledge' in the sense of $n$th-order shared knowledge for some large $n$, leads to a very different expected player behaviour than `absolute common knowledge'. Essentially his conclusion, translated into the context of the mountain trekking example, is that common knowledge is expected to make the two mountain hikers both go to the base, whereas if there is only $n$th-order shared knowledge for some $n$, they are expected to both meet at the top, independent of $n$ and despite the bad weather condition (resulting in non-maximal payoffs). Rubinstein does a pure game-theoretic analysis of optimal strategic behaviour in the game with no experiments, and only speculates what people playing the game might do. This paper aims at elucidating what actual human intuitions command in situations like this. Which depth of shared knowledge (if any) would be enough to attempt risky coordination that would lead to maximal payoff if successful? 

The electronic mail game and the mountain trekking example are complicated in terms of the dynamics of iterated message passing. In this paper, we devise a novel game in which the higher orders of shared knowledge are not achieved dynamically via actions, but are already present at the beginning of the game, using uncertainty about arrival times. This, we believe, makes the game easier to understand. Letting humans play our game, we have been able to address the previous questions in more detail. Our results indicate that the lack of common knowledge does not defer people from being confident in their ability to coordinate, except at the very lowest depths of shared knowledge. So people behave as if having common knowledge, despite only having a low level of shared knowledge, and despite the significant payoff penalties incurred. 
Our conjecture is that the sole presence of shared knowledge is enough to make participants try to coordinate, and that moderate depths of shared knowledge become effectively indistinguishable from common knowledge. 
Our follow-up interviews seem to confirm this conjecture. 
We call this effect ``the curse of shared knowledge'' because even small depths of shared knowledge raises the participant's expectation of being able to coordinate in spite of repeated payoff penalties for having miscoordinated before.

\section{Experimental Design: Materials and Methods}\label{sect:experimental_design} 
For the main experiment, we recruited a total of 680 participants from Amazon Mechanical Turk (MTurk) to play for a maximum of 10 rounds, making a total of $n=4260$ choices. The MTurk experiments were approved by the the Research Ethics Committee at the Faculty of Humanities, University of Copenhagen, Denmark on 22 February 2019 (see Appendix~\ref{app:mturk-walkthrough} for the detailed formulation of the approval). The online experiments were conducted in the period from the 25th February 2019 to 1st March 2019. After accepting
our `Human Intelligence Task' (HIT) and providing written informed consent (see Appendix \ref{app:mturk-walkthrough} for the exact formulation of the Informed Consent Form), participants from MTurk were put in a `waiting room' until they were paired up with another participant to play the game. More detailed descriptions of the MTurk experimental setup can be found in Appendix~\ref{app:mturk-walkthrough}.    

In addition to the main experiment, we conducted three supplementary classroom experiments, with  respectively 80 students (DTU1: $n=2160$) from the Technical University of Denmark (DTU) taking a course on Artificial Intelligence and Multi-Agent Systems, 42 students (DTU2: $n=1012$) taking an introductory course in Artificial Intelligence, and 12 students (DIS: $n=180$) taking a course on Artificial Intelligence at DIS---Study Abroad in Scandinavia (DIS). The two DTU classroom experiments differed slightly from the MTurk experiment in that the students got an initial endowment of \$30 instead of \$10 given to the MTurk participants, and played 30 rounds instead of 10 (see below). Also, in the DTU classroom experiments, all students were told that they would not receive any monetary rewards for playing the game, but they should still try to do their best. The students also had to answer a few additional post-game questions (see  Appendix \ref{app:dtu-experiments} for a full list of those questions). The DIS classroom experiment also differed in the formulation of the game itself, and the way it was presented to the students, see details further below. 

The experiment is designed as a two-player coordination game with imperfect information. The game is inspired by the structure of the consecutive number riddle, also called the Conway paradox, see e.g.\ \cite{van1980conway,van2015one}. Our game is framed as an everyday situation, where two colleagues arrive at their workplace in the morning, and have to decide whether to meet in the canteen for a morning coffee or go straight to their offices and start working immediately. We call the game the `Canteen Dilemma'. The purpose of framing it in an everyday situation is to attempt to make some of the recursive reasoning easier to comprehend~\cite{meijering2010facilitative, wason1971natural}. The introductory story of the game goes as follows: 
\begin{quote}
\indent
``Every morning you arrive at work between $8{:}10$ am and $9{:}10$ am. You and your colleague will arrive by bus 10 minutes apart. Example: You arrive at $\textbf{8:40 am}$. Your colleague may arrive at $\textbf{8:30  am}$, or $\textbf{8:50 am}$. Both of you like to meet in the canteen for a cup of coffee. If you arrive before 9:00 am, you have time to go to the canteen, but you should only go if your colleague goes to the canteen as well. If you or your colleague arrive at 9:00 am or after, you should go straight to your offices.''
\end{quote}

The game has three possible outcomes in each round: 1) both choose the canteen which we refer to as \emph{coordination into the canteen}; 2) both choose their respective offices which we refer to as \emph{coordination into the offices}; 3) one chooses the canteen and the other chooses the office which we refer to as \emph{miscoordination}. At the beginning of each round, participants are told only their own arrival time, and based on this they will have to decide whether to go to the canteen or to the office. After choosing one of those, participants are asked to estimate their certainty that their colleague will make the same choice (on a five-point Likert scale). 

A fixed participation fee of \$2 is given to all MTurk players. Additional bonuses are calculated as penalties subtracted from the initial endowment of \$10 (MTurk) or \$30 (DTU) in each round. Payoffs are tiered in such a way that a small penalty is deducted for successful coordination into the canteen (achieving the highest payoff), which is doubled for coordination into the offices (achieving the second-highest payoff), while the penalty for miscoordination or forbidden choices, i.e., going to the canteen at 9 am or after, is much larger (up to 921 times larger, meaning a significantly lower payoff than the previous two). This means that if a pair of players miscoordinate often, they risk loosing all their endowment, in which case their game terminates prematurely. Penalties are calculated by a logarithmic scoring rule, which depends on the decision and the certainty estimate by each player. Using penalties instead of bonuses may at first seem an unintuitive way to reward participants, but logarithmic scoring rules have been shown to work well due to their ability to ensure that loss minimization remains central, and that the players' actions represent their actual beliefs~\cite{good1992rational, seidenfeld1985calibration, palfrey2009eliciting, mccutcheon2019favor}. 
As noted in Appendix~\ref{appendix:payoffs}, we find a good match between certainty estimates and choices at arrival times different from those that are prone to miscoordination. This corroborates that players tried to minimize their losses and that they made their choices and certainty estimates as honestly as possible. 

We included payoff examples of both successful and failed coordinations in the instructions shown to the participants before the game started, as well as during the game in each round. Screenshots and full descriptions of the experimental setup can be found in Appendix \ref{app:mturk-walkthrough} and details of the payoff structure are provided in Appendix~\ref{appendix:payoffs}. 

The DIS classroom version of the game experiments were different, since we wanted to test the robustness of our findings against variations in the specific formulations of the game and its setup. Additionally, we wanted to address two specific potential criticisms of the original game setup. First of all, one might speculate that the formulation of the game might create a certain behavioural bias, for instance a bias towards going to the canteen due to a personal preference for coffee over office. Second, when doing experiments on MTurk, the players  in principle cannot be guaranteed that their co-players have seen the same rules, or that they are even human. Hence the rules and the game setup might not be common knowledge, as is otherwise expected. In the DIS classroom experiments, the following rules were presented publicly by being shown on the classroom projector and read aloud (see Appendix~\ref{app:dis-rules-full} for the full version): 
\begin{quote}
\indent
The 2 players have to try to coordinate their actions, and if they are successful, they will both achieve the same positive reward [...]. The positive reward can either be \$1 or \$2. If you’re miscoordinated, you lose all the money you earned so far in the game. [...] 
  
In each round, each player is dealt a face down card. On the face of the card there is a number in the interval 2-10 [...]. The numbers on the two dealt cards are always exactly 1 apart, so if for instance one of the players gets a 3, then the other will either be getting a 2 or a 4. [...] 

When the cards are dealt, each player looks at the number on their own card without showing it to the other. The two players are not allowed to communicate or exchange any kind of information during the game. After each player has inspected their own card, they should hide either a white or a black marble stone in their hand and put the hand on the table. [...] 

Both players then receive rewards depending on the colors of their stones and the numbers on their cards. The white marble stone is worth \$2, and the black marble stone is worth \$1, but you only get to keep your stones (money) if you choose the same stone as the opponent, and additionally, if you both choose the white stones, you only get to keep them if both card numbers are strictly below 9. In more details:  
\begin{enumerate}
  \item[a)] If one player chose a white stone and the other a black stone, then they both lose all the stones (money) they have received in the game so far.    
  \item[b)] If both players chose a black stone, they both receive \$1 (i.e., a black stone). 
  \item[c)] If both players chose a white stone, then they both receive \$2 (i.e., a white stone), but only if both numbers are strictly below 9, meaning they are both in the interval 2-8. If not, the players lose all the money (stones) they have received in the game so far.
\end{enumerate}
\end{quote}
This new version of the game is structurally equivalent to the original, but now formulated in a completely abstract way. This potentially makes it harder to grasp the rules, but avoids the potential problem of players making choices affected by everyday personal preferences. By reading up the rules and showing them on the projector in class solves the issue regarding whether the game rules are common knowledge. The reward structure is different in this version of the game. Instead of logarithmic scoring, we chose a simpler scoring with an even higher penalty (lose everything) for miscoordination (different stones) and ``forbidden'' choices (choosing white when having 9 or 10). This might help making it more clear to the players that choosing white is never worth the risk of miscoordination. 

The overall point of both versions of the game is that the only `safe' strategy is to always coordinate on the choice that gives the lower payoff (going to the offices or choosing black stones). Even if it seems safe to try to coordinate for the higher payoff when arriving sufficiently ahead of 9:00 am or getting a card with a number sufficiently below 9, it is never common knowledge that both arrived before 9:00 am or both have a card below 9, which makes it unsafe to attempt coordinating for the higher payoff.


\section{Game Strategies}
What are the relevant strategies for this game? First note that going to the canteen at $9{:}00$ or after results in the worst possible payoff. So both players should always go to the office if they arrive at $9{:}00$ or after. How about if both arrive strictly before $9{:}00$? If both choose canteen, they getter a better payoff than if both choose office. Now consider a case where you are one of the players, and you arrive at $8{:}50$. Then your colleague will be arriving at either $9{:}00$ or $8{:}40$. If your colleague arrives at $9{:}00$, she has to choose office according to the previous argument, and then you would have to choose office as well to avoid the large penalty of miscoordination. However, if your colleague arrives at $8{:}40$, you may both choose the canteen, and this will lead to the highest payoff. In other words, depending on the arrival time of your colleague, a piece of information that you don't have access to, the best choice is either office or canteen. So which one to choose?

Since the penalty of miscoordination is very high, it would seem best to choose office. What if you then instead  arrive at $8{:}40$? In this case, your colleague either arrives at $8{:}30$ or $8{:}50$. In both cases, you have time to meet for a cup of coffee in the canteen, and doing so will give you the highest payoff. At first, it might seem like an easy choice. However, we just concluded that the best strategy at $8{:}50$ would be to go to the office. So, if you arrive at $8{:}40$ and contemplate that your colleague might arrive at $8{:}50$---and if you believe your colleague would reason as yourself and go to the office at $8{:}50$---you also ought to go to the office at $8{:}40$. This argument can of course be iterated, because if the optimal choice at $8{:}40$ is to go to the office, then the optimal choice at $8{:}30$ must also be to go to the office. In other words, the optimal strategy seems to be to always go to the office, independent of arrival time! And, indeed, so it is. If both players go to the office in all rounds and declare the highest possible certainty in their decisions, they will both leave the experiment with \$9.80, excluding the \$2 participation fee (see again Appendix~\ref{appendix:payoffs} for details of the payoff structure). This is the highest possible payoff that can be guaranteed by any strategy in the game, and very close to the \$10 that the players start out with. As we will see later, the payoffs that people actually get when playing the game are \emph{significantly} lower than this.

The \emph{all-office} strategy described above, where you always decide to go to the office independent of arrival time, is a safe strategy if both players follow it. By safe is meant that there is never any risk of miscoordination, and hence no risk of getting the highest penalty (the penalty for miscoordination is up to \$9.21 in a single round, see Appendix~\ref{appendix:payoffs}). It is actually the \emph{only} safe strategy. The reason is that if at least one of the players, say $a$, has the strategy of going to the canteen at some time $t$ before $9{:}00$, then since they both have to go to the office at $9{:}00$ or later, there must exist at least one pair of arrival times for which the two players are miscoordinated.\footnote{
\label{footnote:no-cutoff}  
We now prove this claim. Suppose $a$ has the strategy of going to the canteen at some time $t$ before $9{:}00$, and that both $a$ and $b$ will go to the office if they arrive at $9{:}00$ or later. To derive a contradiction, suppose $a$ and $b$ are always coordinated, i.e., always make the same choice for any pair of arrival times. Since $a$ chooses to go to the canteen at time $t$, player $b$ also has to go to the canteen at time $t+0{:}10$, since otherwise whey would be miscoordinated when $a$ arrives at $t$ and $b$ at $t+0{:}10$. But if $b$ goes to the canteen at time $t+0{:}10$, $a$ also has to go to the canteen at time $t+0{:}20$, since otherwise they would be miscoordinated when $b$ arrives at $t+0{:}10$ and $a$ at $t+0{:}20$. This can be generalized to conclude  that $a$ would have to go to the canteen at any time $t+0{:}20x$ for $x \geq 0$ and $b$ would have to go to the canteen at any time $t+0{:}10+0{:}20y$ for $y \geq 0$. Clearly this implies going to the canteen after $9{:}00$, which is a contradiction.}

The fact that the all-office strategy is the only safe one is counter-intuitive to most people before being presented with the proof, and for some people even after. The issue is that, intuitively, it would seem to be safe to go to the canteen at, say, $8{:}30$. Why would a person/player, say Agnes, ever go to the office that early? She knows that her co-player/colleague, say Betram, will then be arriving at latest at $8{:}40$, which still leaves plenty of time to get a cup of coffee together before $9{:}00$. The issue is of course that if Agnes takes the perspective of Bertram, then assuming he arrives at $8{:}40$, he will consider it possible that she arrives at $8{:}50$. And if she had indeed arrived at $8{:}50$, she would consider it possible that he arrives at $9{:}00$. In that case she would be forced to choose the office. A major point of our experiments is to test whether this kind of recursive perspective-taking is utilized by human players of the game.

The argument of the all-office strategy being safe of course relies on the other player following the same strategy. Since we don't allow players to agree on a strategy with their co-player beforehand, the all-office strategy doesn't necessarily in practice lead to the highest payoff for a particular player. Another issue is that one might decide to play risky instead of safe. Consider the \emph{canteen-before-9} strategy of always going to the canteen before $9{:}00$ and going to the office at later times, all with the highest certainty estimate. If both players choose this strategy and are fortunate to play 10 rounds without any of them arriving  at $9{:}00$ or later, they will get the highest possible payoff of \$9.90---slightly higher than the guaranteed payoff \$9.80 of the all-office strategy. However, if all pairs of arrival times are equally likely, the probability of miscoordination is then $1/6$ (there are 12 pairs of arrival times in total, and 2 of those have one player arriving at $8{:}50$ and the other at $9{:}00$). Miscoordination with the highest possible certainty estimate gives a penalty of \$9.21, so in practice this strategy is of course still significantly worse than the all-office strategy, even if only playing one round (the expected payoff of the canteen-before-9 strategy in a single-round game will be $\$10.00 - \$9.21\cdot 1/6 - \$0.01\cdot 5/6 = \$8.46$, whereas the all-office strategy guarantees a payoff of $\$10.00-\$0.02 = \$9.98$, cf.\ Appendix~\ref{appendix:payoffs}).

In Appendix~\ref{appendix:formal}, 
we make the reasoning about game strategies formally precise. We show that independently of the particular payoff structure (only using the relative order of the payoffs, not their exact values), there will only be two candidates for the optimal strategy, the all-office strategy and the canteen-before-9 strategy. Which one is then optimal depends on the particular payoff structure and the number of arrival times before and after 9 am. In our specific experiments, with our specific arrival times and payoff structure, the all-office strategy has a significantly higher expected utility than the canteen-before-9 strategy, as already argued. As we will see, the human players in our experiments very rarely play any of these strategies, but seem to believe that it is safe to go to the canteen 
when arriving sufficiently ahead of 9 am (e.g.\ before $8{:}50$), but unsafe when arriving later (e.g.\ at $8{:}50$ or later).

A strategy to always go to the canteen if arriving before some cut-off time $t_c$ and always go to the office if arriving at or after $t_c$ is called a \emph{cut-off strategy} (with cut-off $t_c$). The canteen-before-9 strategy is a cut-off strategy with cut-off $9{:}00$ (see Appendix~\ref{appendix:formal} for more details).

\section{Results}
\label{S:results}
\begin{table}
\caption{} 

\smallskip
\centering
\setlength{\arrayrulewidth}{1pt}
\setlength{\tabcolsep}{4pt}
\begin{tabular}{l|ccccccc}
\hline
\small Treatment   & \small   N &   \small MaxRounds &  \small  avg. rounds &  \small  n &  \small Ruin & \small   payoff & \small   Avg penalty \\
\hline
\small AMT         & 680 &          10 &          6.3 & 4260 &  52.79 &    23.59 &         -1.59 \\
\small DTU1        &  80 &          30 &         27    & 2160 &  17.5  &    27.03 &         -0.83 \\
\small DTU2        &  42 &          30 &         24.1  & 1012 &  30.95 &    24.11 &         -0.98 \\
\hline  \multicolumn{7}{c}{} \\[-2mm]
\multicolumn{8}{p{12cm}}{Aggregate statistics of all three experiments. $^*$N = number of subjects; pairs = number of pairs; R = maximum number of rounds; $\bar r$ = average number of rounds played; n = number of choices; Ruin = percentage of participants loosing all their bonus before (or in) round R; Payoff = average earnings (given as the retained percentage of the initial endowment); $\bar s$ = average penalty per player per round.} \\
\end{tabular}
\label{table:1}
\end{table}
The maximal theoretical payoff described in the previous section was never observed in the experiments---actually quite far from it, despite doing the experiment with more than $800$ people. Recall that the payoff of the all-office strategy is $\$9.80$ independent of arrival times. The average bonus paid to our MTurk participants was a mere $\$2.36$. Due to the penalty-based payoff structure, only 46 out of $340$ MTurk pairs ($14\%$) were able to play $10$ rounds and still have any bonus left, while the average number of rounds played was $6.3$, see Table~\ref{table:1}. As soon as one of the players had no money left, the game would terminate.

Comparing the MTurk experiment with the DTU experiments in Table~\ref{table:1}, shows that the DTU participants were slightly better on average. While more than half of the MTurk participants had lost their initial bonus and had to end the game before the last round, only $18\%$ and $31\%$ of the DTU participants, respectively, had done so. Especially the students from the Artificial Intelligence and Multi-Agent Systems course (DTU1) managed well by retaining $27\%$ of the initial endowment and loosing only $\$0.83$ per round on average. 

\begin{figure} 
	\centering\includegraphics[width=0.8\linewidth]{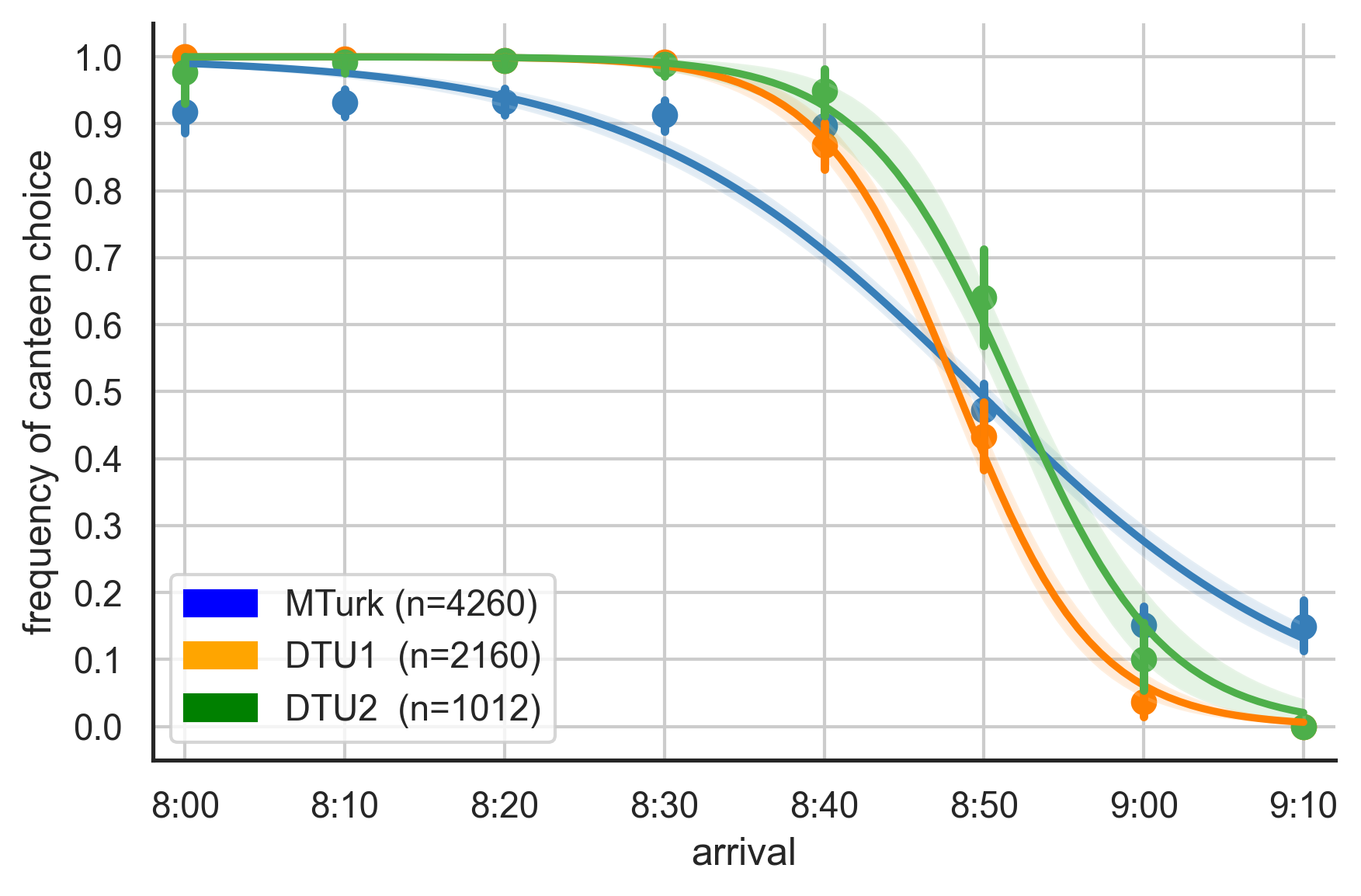}
	\caption{Percentage of canteen choices as a function of arrival times. The colored lines are logistic regression lines with 95\% confidence intervals shown as translucent bands. Fitted parameters show significant differences in the slope and intercept between MTurk and DTU experiments ($p < .0001$).}
	\label{fig:1}
\end{figure}

To model the relationship between arrival time and the probability of choosing the canteen (choice), a logistic regression is performed in Fig.~\ref{fig:1}. The model computes a binomial model with bootstrapped confidence intervals (CI $=95\%$) for the estimated parameters. Bootstrapping is performed with $10.000$ iterations to ensure stable estimation of confidence intervals. MTurk participants in blue have a slightly more gradual decline of canteen choices for increasing arrival times. However, the point at which there is a $50\%$ probability of choosing the canteen or the office is close to $8{:}50$ in all three experiments (see further specfications in an additional mixed-effects regression analysis in Appendix~\ref{appendix:logistic}). Later, in Fig.~\ref{fig:indistinguish}, we make a combined regression analysis of all three experiments in order to understand the overall experimental results in terms of degrees of shared knowledge.

\begin{figure} 
\centering\includegraphics[width=0.8\linewidth]{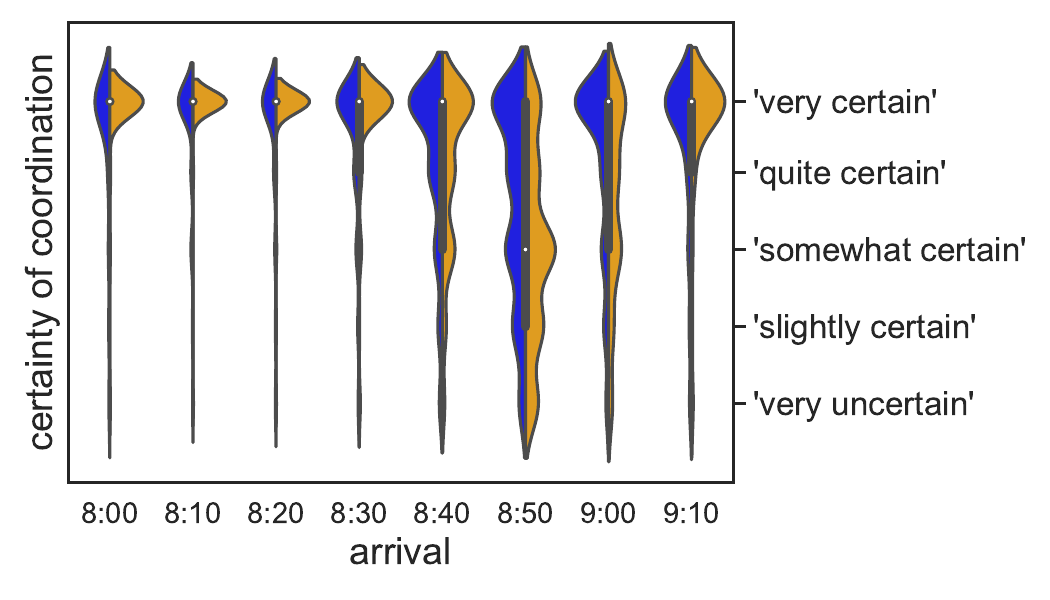}
\caption{Violin plots of certainty estimates. In each round, participants were asked how certain they were of successful coordination with their colleague. Blue areas show the results from MTurk ($\text{n}=4260$) and orange areas show the results from DTU1 and DTU2 combined ($\text{n}=3172$). We predefined a five point likert scale of certainty estimates as: `very uncertain', `slightly certain, `somewhat certain', `quite certain', and `very certain', and translated them into the numerical values of probability estimates used in the payoff calculations (see Appendix~\ref{appendix:payoffs}). The white dots correspond to the median certainty estimate.}
\label{fig:certain}
\end{figure}

The violin plots in Fig.~\ref{fig:certain} show the distribution of certainty estimates for each arrival time. The DTU experiments are merged and compared with the MTurk experiment, as the differences between the two DTU experiments generally are smaller than their respective difference to the MTurk experiment. This is quantified in Table~\ref{table:1}, as the calculated average penalties per round are a direct consequence of the certainty estimates given by the participants. In general, however, Fig.~\ref{fig:certain} shows that it is exceedingly rare for any of the participants to consider it problematic to go to the canteen when arriving early. Arriving at $8{:}30$ or earlier is deemed sufficiently early to visit the canteen with very high confidence. 
The subtle differences in certainty estimates between MTurk participants and DTU students show that the latter tend to be slightly more certain that their co-players follow a similar strategy (higher certainty estimates for the early and late arrival times), and also that they are more aware of the danger of miscoordination (lower certainty estimates at $8{:}50$ and $9{:}00$). This is in particular the case for the DTU1 experiment that has the steepest profile. Being more certain that your co-players follow a similar strategy probably indicates that you believe such a strategy to be optimal. So, interestingly, the DTU1 participants are both the ones that appear to be most aware of the danger of miscoordination, and at the same time those who most firmly believe a cut-off strategy is optimal, i.e., believing that the risk of miscoordination is unavoidable. The differences between the three experiments are however still relatively minor. In the following we will combine data from all three experiments.

\begin{figure} 
\centering\includegraphics[width=0.8\linewidth]{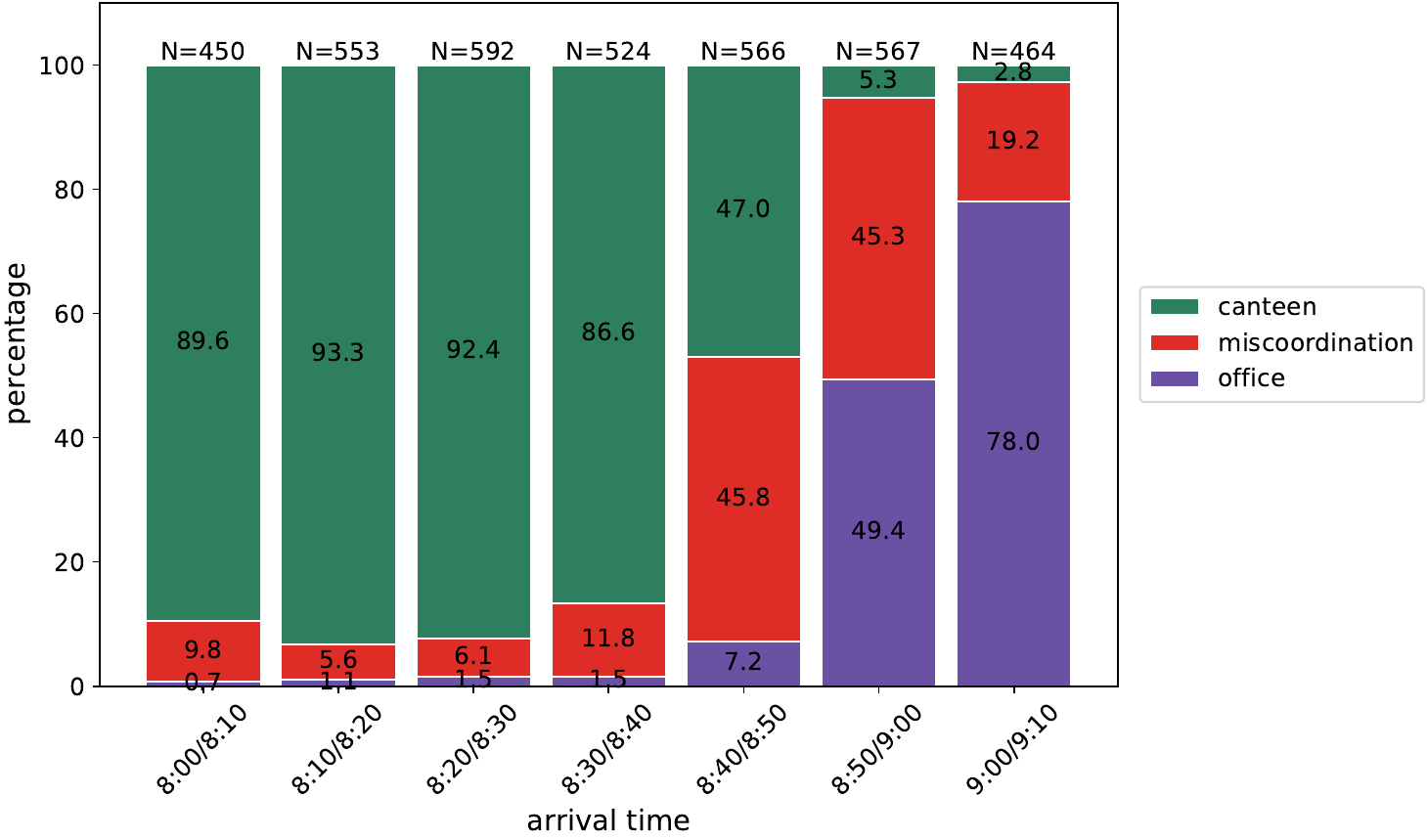}
\caption{Percentage of coordinations and miscoordinations as a function of arrival times. Green means coordinating into the canteen, purple means coordinating into the office, and red means miscoordination. We use the notation 8:00/8:10 to denote the union of the arrival pairs $(8{:}00,8{:}10)$ and $(8{:}10,8{:}00)$, i.e., the arrival time combinations where one of the players arrive at $8{:}00$ and the other at $8{:}10$. Miscoordinations approach $50\%$ at 8:40/8:50 and 8:50/9:00.}
\label{fig:miscoordinations}
\end{figure}

In Fig.~\ref{fig:miscoordinations}, we see the percentage of successful coordinations into the canteen or office (green or purple) together with the number of miscoordinations (red) as a function of all possible arrival time combinations. The figure shows that players are able to coordinate into the canteen more than 86\% of the time if both of them arrive before $8{:}50$. As soon as a pair has a player who arrives at $8{:}50$, however, the result changes drastically. Suddenly almost half of such pairs miscoordinate. As players experience harsh penalties for miscoordinating, one could perhaps expect to see a tendency of choosing office more often when arriving at 8:40 or 8:50 in subsequent rounds. That is, we might expect that players learn and converge to the all-office strategy in order to avoid miscoordination altogether. But this is not what we see.

\begin{figure} 
\centering\includegraphics[width=0.8\linewidth]{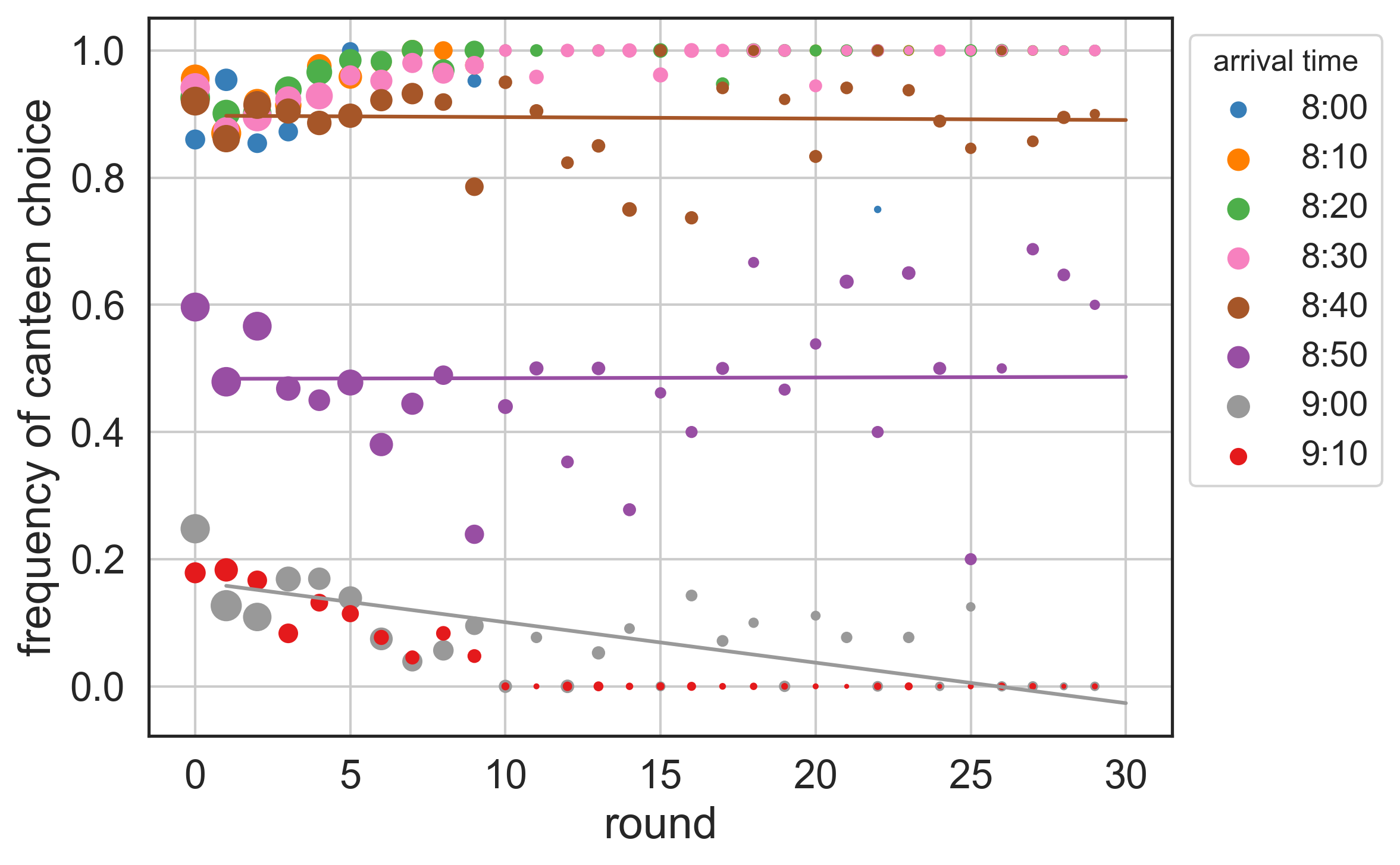}
\caption{Mean frequencies of canteen choices for all possible arrival times as a function of the number of rounds played. The three fitted horizontal lines, corresponding to arrival times $8{:}40$, $8{:}50$, and $9{:}00$ respectively, represent weighted linear squares (WLS) with the weights chosen to be the square root of the number of data points constituting the mean frequencies for each round, also shown by dot size.}
\label{timeseries}
\end{figure}
Fig.~\ref{timeseries} shows the mean frequency of canteen choices as a function of rounds played for all three experiments. Each color corresponds to a certain arrival time. Clearly, the only arrival times that do not converge towards either the canteen or the office are the arrival times of $8{:}40$ and $8{:}50$, with the former fluctuating around $90\%$ canteen choices and the latter fluctuating around $50\%$ canteen choices. This indicates that there is no behavioral change during the game. Participants arriving at $8{:}40$ or $8{:}50$ do not feel incentivized to change their behavior significantly in subsequent rounds, even though there is a high risk of miscoordination. This is not to say that participants do not learn that canteen choices at $8{:}40$ or $8{:}50$ are dangerous. Partitioning the data from Fig.~\ref{fig:1} into two bins, corresponding to pairs having had no miscoordination and pairs having had one or more miscoordinations (see the supplementary data analysis in Appendix~\ref{app:supplementary-data-analysis}), shows somewhat decreasing certainty estimates around the critical arrival times. However, this does not affect their actual choices. MTurk participants do choose the canteen a little less often after a miscoordination (see Fig.~\ref{fig:certainties} in the appendix), but this does not translate into better payoffs as later miscoordinations just move to earlier arrival times. So even though participants learn that their choices are risky, they do not see any way to improve their strategy (at least not in the short term). Specifically, they never converge to the optimal all-office strategy, and also not to the alternative canteen-before-9 strategy (cf.\ Theorem~\ref{theorem:all-office-or-cut-off} in Appendix~\ref{appendix:formal}). This apparent lack of behavioral change in higher-order social reasoning games was also identified by Verbrugge~\cite{verbrugge2008learning}. 

\begin{figure} 
  \centering\includegraphics[width=0.6\linewidth]{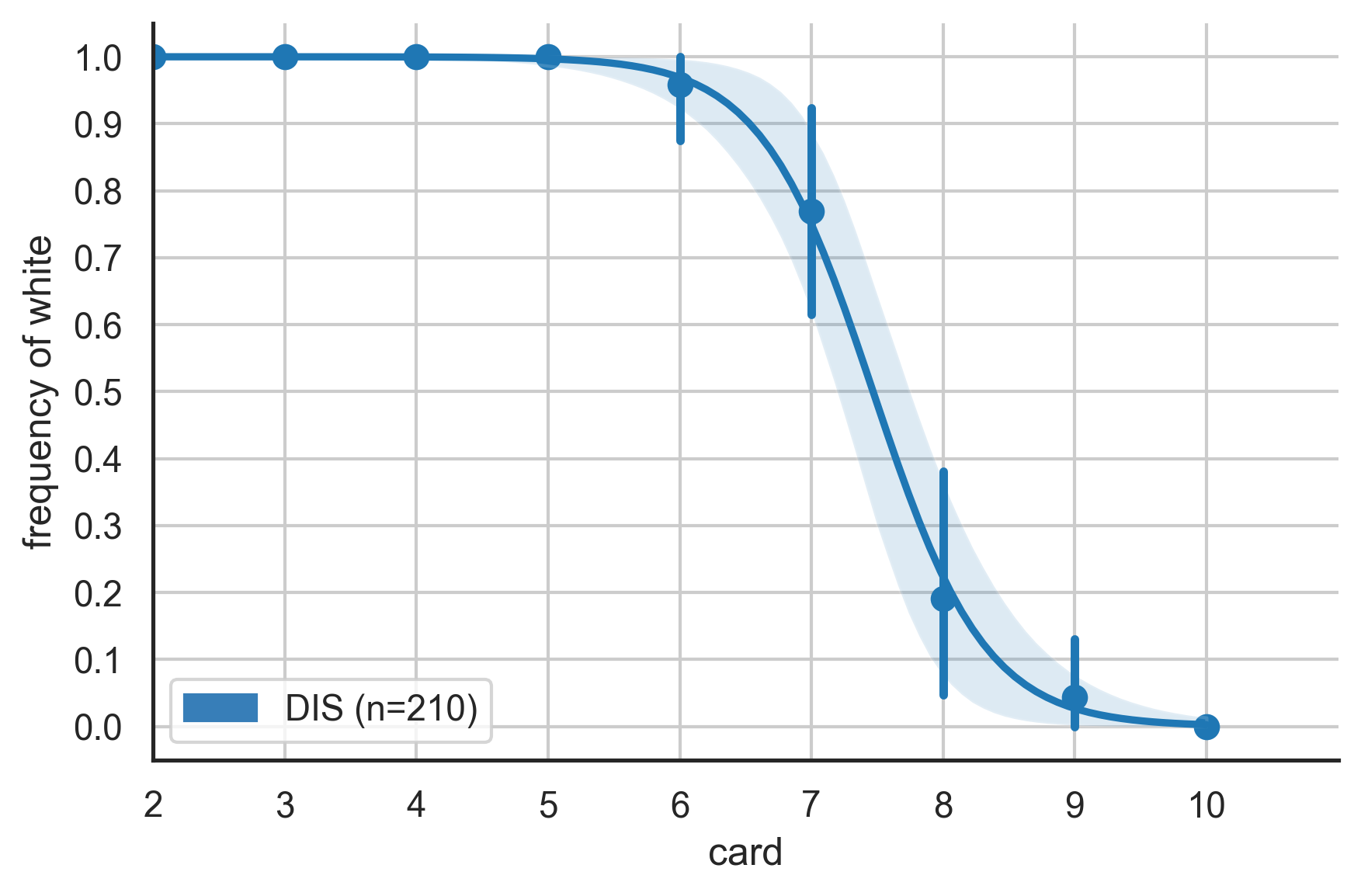}
  \caption{Logistic regression of the relationship between card number and choice of a white stone in the qualitative experiment at DIS. The resulting curve shows a very similar overall trend as in the main experiments shown in Fig.~\ref{fig:1}.} 
  \label{fig:dis} 
  \end{figure}  
Fig.~\ref{fig:dis} shows the results for the alternative version of the game conducted at DIS, where players are dealt cards with numbers from 2 to 9, and choose either a black or a white stones (cf.\ Section~\ref{sect:experimental_design}). Choosing a white stone corresponds to choosing office in the original version of the game, and getting a card with the number $n$ corresponds to arriving at $9{:}00+00{:}10(n-9)$ in the original game. So e.g.\ receiving card 7 correponds to arriving at $8{:}40$, receiving card 8 corresponds to arriving at $8{:}50$, and receiving card 9 corresponds to arriving at $9{:}00$. The results in Fig.~\ref{fig:dis} are very similar to the results of the original version of the game, Fig.~\ref{fig:1}. The curve has the same shape, and with a slope that is most similar to the DTU experiments (which makes sense, as here the participants are also students in computer science and AI). The choices at 6 ($8{:}30$) are more than $90\%$ white (canteen), and the choices at 9 ($9{:}00$) are more than $90\%$ black (office), similar to the DTU1 experiments. The choices at 8 ($8{:}50$) are around $20\%$ white (canteen), which is lower than for any of the original experiments, suggesting a potentially better understanding of the risk of miscoordination. Alternatively, it might be due to the fact that in this version of the game, miscoordination leads to immediately losing all earnings, so the penalty of miscoordination is even clearer and higher in this version of the game (this was a deliberate choice in our game design, in order to make sure that the observed results of the original game were not due to unawareness of the high penalty of miscoordination).

The DIS classroom experiment was followed by a number of qualitative interviews with the students, divided into groups of 4.  These confirmed that the players considered the rules to be common knowledge. 
When asked about the existence of safe strategies guaranteeing coordination (always choosing the same stone), 
most players expressed the false belief that a cut-off strategy existed that would guarantee coordination on white stones for sufficiently low numbers (which we proved false in footnote~\ref{footnote:no-cutoff}). Most players believed that the cut-off would be around 6 or 7,  
although some realized during the interview that such a  strategy can never guarantee coordination. Only one student was skeptical about the existence of a safe strategy for coordinating on white from the very beginning of the interview.

\section{Discussion}\label{discussion}
Let us try to analyse the experimental results in terms of the depth of knowledge of the participants. The highest payoff is achieved when successfully coordinating into the canteen before 9:00. With the aim of achieving the highest possible payoff, each participant can be expected to consider her own arrival time and try to assess whether there is still time to meet in the canteen. When a participant arrives strictly before 9:00, i.e.\ at $8{:}50$ or earlier, she has private knowledge that she arrives sufficiently early to go to the canteen. If participants only make choices based on their private knowledge, we should then expect participants to always go to the canteen at $8{:}50$. This is not what we see, cf.\ Fig.~\ref{fig:1}. Thus, other considerations in addition to the player's private knowledge must play a role in their decision making.

When both participants know they arrived before 9:00, they have shared knowledge of having arrived in time for going to the canteen. This happens for any arrival pair $(t_1,t_2)$ with $t_i \leq 8{:}50$, $i=1,2$ (where an \emph{arrival pair} $(t_1,t_2)$ denotes that player $1$ arrives at time $t_1$ and player $2$ at time $t_2$, see Appendix~\ref{appendix:formal} for further details). Note that for an arrival pair $(8{:}50,8{:}40)$, there is shared knowledge of there being sufficient time to go to the canteen, but only player $2$ knows this fact: Player $2$ knows that also player $1$ must have arrived before 9{:}00, but player $1$ doesn't know this about player $2$. In other words, when a player arrives at $8{:}50$, that player considers it possible that there is shared knowledge of there being sufficient time for a cup of coffee in the canteen, but only if arriving at $8{:}40$ or before will that player \emph{know} there to be shared knowledge (to depth 1). When arriving at $8{:}30$ or before, the player additionally knows there to be shared knowledge to depth 2. We illustrate this in Fig.~\ref{fig:indistinguish}. Note that in general, if a player arrives at time $8{:}50-0{:}10n$, $n>0$, then that player knows that there is $n$th-order shared knowledge, but the player doesn't know there to be $(n+1)$st-order shared knowledge. This follows a similar pattern as the mountain trekking example, except here the depth of shared knowledge is determined by how early ahead of 9:00 the agents arrive, rather than how many messages have successfully been delivered. No number of messages was sufficient to achieve common knowledge in the mountain trekking example. We similarly get that no arrival time is sufficiently early to establish common knowledge of having sufficient time to meet for coffeee in the canteen. 

\begin{figure} 
\[
  \begin{tikzpicture}[align=left,xscale=0.9]
   \node[anchor=west] at (-1.0,-2.5) {players' \\ knowledge level}; 
   \node[anchor=west,red] at (-1.0,0.9) {arrival time \\ player 1:};
   \node[anchor=west,mygreen] at (-1.0,-1) {arrival time \\ player 2:};
   \node[red] (n1820) at (3,1) {8{:}20};
   \node[mygreen] (n2820) at (3,-1) {8{:}20};
    \node[red] (n1830) at (5,1) {8{:}30};
   \node[mygreen] (n2830) at (5,-1) {8{:}30};
   \draw[-] (n1820) to (n2830);
   \draw[-] (n2820) to (n1830);
  \draw[dashed] (4,-3) to (4,1.5);
    \node[red] (n1840) at (7,1) {8{:}40};
   \node[mygreen] (n2840) at (7,-1) {8{:}40};
    \node[red] (n1850) at (9,1) {8{:}50};
   \node[mygreen] (n2850) at (9,-1) {8{:}50};
     \draw[-] (n2840) to (n1850);
   \draw[-] (n1840) to (n2850);
   \draw[-] (n1830) to (n2840);
   \draw[-] (n2830) to (n1840);
  \draw[dashed] (6,1.5) to (6,-3);
  \draw[dashed] (8,1.5) to (8,-3);
     \node[red] (n1900) at (11,1) {9{:}00};
   \node[mygreen] (n2900) at (11,-1) {9{:}00};
    \node[red] (n1910) at (13,1) {9{:}10};
   \node[mygreen] (n2910) at (13,-1) {9{:}10};
     \draw[-] (n2850) to (n1900);
   \draw[-] (n1850) to (n2900);
   \draw[-] (n1900) to (n2910);
   \draw[-] (n2900) to (n1910);
  \draw[dashed] (12,1.5) to (12,-3);
  \draw[dashed] (10,1.5) to (10,-3);
  \node at (13,-2.5) {none};
  \node at (11,-2.5) {none};
  \node at (9,-2.5) {private};
  \node at (7,-2.5) {shared\\ depth 1};
  \node at (5,-2.5) {shared\\ depth 2};
 \node at (3,-2.5) {shared\\ depth 3};
 \node at (7.8,0)
 	{\includegraphics[scale=.75]{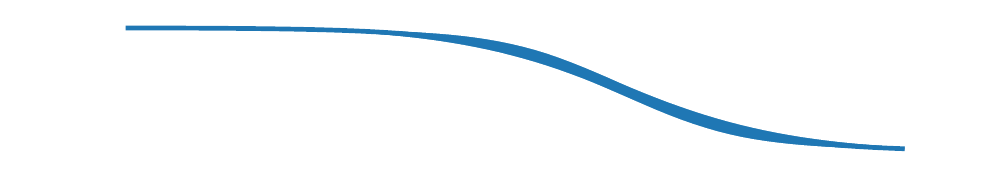}};
  \end{tikzpicture}
  \]
  \caption{The solid diagonal lines express indistinguishability for the players, e.g.\ the arrival time $8{:}40$ for player $1$ has a line to both of the arrival times $8{:}30$ and $8{:}50$ for player $2$, since these are the two arrival times for player $2$ that player $1$ will consider possible when herself arriving at $8{:}40$. Below each possible arrival time, we have marked the highest level of knowledge of there being sufficient time to go to the canteen, e.g.\ when arriving at $8{:}40$, there is shared knowledge to depth 1 of this fact, but not shared knowledge to depth $2$. In blue, a binary logistic regression model was used to predict the probability of a participant going to the canteen (upper limit) or to the office (lower limit) at the shown arrival times. The width of the regression line indicates the 95\% confidence interval using 10.000 bootstrapped resamples of all choices in all three experiments ($\text{n}=7432$).
  }\label{fig:indistinguish}
\end{figure}
The participants seem to clearly be able to distinguish between private and shared knowledge, which is supported by their significantly different choices at $8{:}50$ and $8{:}40$ (see again Fig.~\ref{fig:indistinguish}). However, it is less clear whether they are able to robustly distinguish different levels of shared knowledge, and whether they are able to distinguish that from common knowledge. Indeed, most participants relatively robustly choose the canteen at $8{:}40$ and any time before that, despite the difference in depth of shared knowledge in those possible arrival times. The certainty estimates are however slowly decreasing from $8{:}10$ to $8{:}50$ in all three experiments (see Fig.~\ref{fig:certain}), showing that the participants are not completely ignorant to the differences. This could suggest that many participants believe that it is less safe to go to the canteen based on $n$th-order shared knowledge than $(n+1)$st-order shared knowledge. Very few seem to draw the conclusion that it is never safe to go to the canteen, or, if they do, they at least don't expect the other player to be able to draw the same conclusion. 

Why do participants not regard earlier office choices as viable options? Why do participants not continue their train of thought and deduce that when $8{:}50$ turns out to be unsafe, $8{:}40$ will become unsafe as well, which means that $8{:}30$ will also be unsafe, etc.? One reason may be that the all-office strategy is cognitively difficult to comprehend given the  limited  ability of humans to take the perspective of each other recursively. 
One might also speculate that the players suffer from a 
computational overload from the complexity of the game's rules and payoff structure. While this may be a problem, we do not think it has much significance for the following reasons. If computational overload was sufficient to distract our participants, we might expect more varied suboptimal choices. While we do see some bluntly irrational choices of canteen at 9:00 (given the rules and payoff structure), the vast majority choose office at 9:00 and 9:10, and canteen at earlier arrival times, with many still choosing office at 8:50. Besides this, we see participants be much less certain of their canteen choices at 8:50 or 8:40 than at earlier arrival times. 

Another reason for not considering earlier office choices as viable may be that while the optimality of an all-office strategy are understood, the participants do not believe that this is common knowledge, and hence they do not rely on the other participants to choose the same strategy. That is, they either think their co-player does not realize there is no safe cut-off, or that their co-player thinks this about them, and so on. In this way, the strategic reasoning involved becomes akin to the Keynesian beauty contest~\cite{keynes1936general,nagel1995unraveling}. Under this assumption, we should not necessarily expect participants to play an all-office strategy, since that is only proven optimal when assuming there to be common knowledge of perfect rationality (see Appendix~\ref{appendix:formal}). A possible alternative strategy may be the following mixed strategy: 1) always go to the canteen before $8{:}50$, 2) always go to the office after $8{:}50$, and 3) go to the canteen with probability $0.5$ at $8{:}50$. If both players follow this strategy, the arrival time combinations $8{:}40/8{:}50$ and $8{:}50/9{:}00$ will then coordinate $50\%$ of the time, matching well with what we observe in Fig.~\ref{fig:miscoordinations}. If this is the strategy followed, and if players assume their co-players to follow the same strategy, they should also be aware of the $50\%$ probability of miscoordination at $8{:}50$. In Fig.~\ref{fig:certain}, we indeed see a much lower coordination certainty estimate at 8:50. However, if players were indeed playing this strategy and assuming their co-players to do the same, their certainty estimates ought to be even much lower at 8{:}50. 
 These certainty estimates also ought to be lower if we assume that players are aware of the possibility of their co-players not being perfectly rational and their chosen strategy hence to be somewhat unpredictable.

\begin{figure} 
\centering\includegraphics[width=1\linewidth]{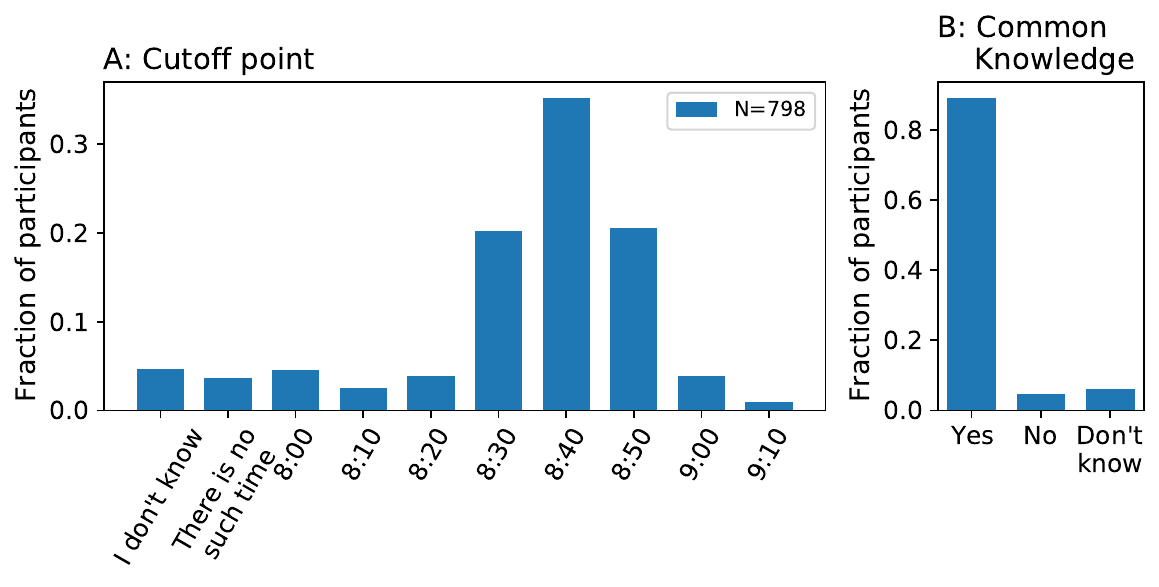}
\caption{A) Frequencies of answers to the question: ``Imagine you could have agreed beforehand with your colleague about a point in time where it is safe to go to the canteen. What time would that be?'' Due to the pragmatics of language, we assume that an answer like 8:30 entails the belief that all earlier arrival times would also be deemed safe. B) Frequencies of answers to the question: ``Imagine you arrive at 8:10. Is it common knowledge between you and your colleague that it is safe to go to the canteen, that is, that you both arrived before 9:00?''}
\label{fig:cutoff answers} 
\end{figure}
Probing theses questions further, we asked participants the following post-game question:
\begin{quote}
\indent
``Imagine you could have agreed beforehand with your colleague about a point in time where it is safe to go to the canteen. What time would that be?'' \textit{('I don't know', 'There is no such time', 8:00, 8:10, 8:20, 8:30, 8:40, 8:50, 9:00, 9:10)}
\end{quote}
The results in Fig.~\ref{fig:cutoff answers}A show that most answers range from 8:30 to 8:50 (approximately $75\%$ of all answers), giving support to the verdict that participants are not able to continue taking the perspective of each other recursively, or at least that they believe that shared knowledge of some modest finite depth is sufficient for the canteen choice to be safe. Rather, they stop after one or two, possibly three, iterations, thus believing that as long as they arrive sufficiently early, they can be sure to coordinate safely in the canteen. Notice that the correct answer, ``there is no such time'', is chosen by less than $4\%$ of all participants. The interviews with students in the DIS experiment indicates the same. One student claims for example that ``(from) two to five, there is no risk involved''. This also reflected in the results, as Fig.~\ref{fig:dis} shows, no students plays black (equivalent to an `office' choice) for numbers 5 and below. One interview from the DIS experiment also ends with a student asking ``Did you find an optimal number?''. 

The answers to the post-game question above, as well as the DIS interviews, indicates that the prevalence of canteen choices is not due to uncertainty about the rationality of other players, but instead due to a more immediately available belief, that there is a point where canteen choices (or white stones in the DIS experiment) can be played without the risk of miscoordination. In reality, they would need common knowledge in order for the canteen choice to be without risk, and the players seem to behave as if they have such common knowledge. 
Participants might of course not necessarily have a precise idea of the technical notion of common knowledge and how it differs from $n$th-order shared knowledge, but as discussed in the introduction, there is actually quite a number of studies demonstrating that humans have adapted to recognize common knowledge and making distinct strategic choices depending on whether there is common, shared or private knowledge---at least in cases where the difference between these states of knowledge is relatively clear (for instance when the common knowledge situation is achieved via a public announcement). In our experiments, we see the player behavior stabilizing already at relatively modest depths of shared knowledge, both in terms of action choices and certainty estimates. We also see that player behavior does not match what we would expect to see if they only believed to have $n$th-order shared knowledge, but instead matches what we would expect to see if they indeed wrongly inferred common knowledge. 

To specifically address the issue of whether they wrongly infer common knowledge, we asked a final post-game question:

\begin{quote}
\indent
``Imagine you arrive at $8{:}10$ am. Is it common knowledge between you and your colleague that it is safe to go to the canteen, that is, you both arrived before $9{:}00$ am?''. \textit{(`Yes', `No', `Don't know')}
\end{quote}

\noindent
This question inquires about participants' understanding of the term `common knowledge', and how it applies to the given situation. In Fig.~\ref{fig:cutoff answers}B, the results show that $89\%$ of all participants responded that it was common knowledge that both players arrived before $9{:}00$ when they themselves had arrived at $8{:}10$. The answers may signify that indeed they believe there to be common knowledge in the strict technical (logical) sense. But of course the answers could also pertain to the everyday linguistic usage of the term `common knowledge', which is less strict.

\section{Conclusion and future work} \label{S:conclusion}
We developed a new coordination game, the \emph{Canteen Dilemma}, to examine higher-order social reasoning in humans. Our experimental results 
indicate a substantial degree of miscoordination, even in scenarios where coordination \emph{is} achievable (albeit only for the second-highest payoff). This miscoordination appears to stem from what we term the ``curse of shared knowledge'': an illusion of common knowledge in situations where only shared knowledge exists, and only to a limited depth. While the game results alone do not conclusively demonstrate that this illusion of common knowledge occurs, they show that players behave as though they have common knowledge. Moreover, variations to the game and responses to post-game questions strongly suggest the presence of this illusion. For instance, an overwhelming majority of participants report having common knowledge about the safety of arriving at the canteen at $8{:}10$ (Fig.~\ref{fig:cutoff answers}B). Additionally, qualitative interviews following the DIS classroom experiment further supported this interpretation. 

While we cannot claim to have definitely proven the existence of a 'curse of shared knowledge', our experiments support it as a plausible conjecture that warrents further investigation. This result is somewhat surprising, as it initially appears to contrast with existing literature on shared and common knowledge. 
Previous findings~\cite{lee2010rationales, thomas2014psychology, thomas2016recursive, thomas2018common, de2019common, de2019maimonides} suggest that people make strategically different choices depending on whether they have shared or common knowledge. 

The key difference between our study and these previous works lies in the experimental designs: Earlier studies explicitly inform participants about their level of knowledge---whether it is shared or common---before asking them to make decisions. In contrast, our study requires participants to deduce their level of knowledge independently, without explicit clarification. 
Thus, while prior research investigates 
decision making when knowledge levels are explicitly stated, our work examines whether individuals can distinguish between shared and common knowledge when this distinction is not made explicit. 

Our findings suggest they cannot; at least not in our game. It is possible that 
specific aspects of our 
experimental design make it particularly challenging to distinguish between shared and common knowledge. 
However, this raises a broader question: If people make different strategic choices when being explicitly informed about 
their knowledge level in controlled experiments, 
does this have implications for real-world decision making 
where such explicit clarification is unlikely or at least probably relatively rare?

For individuals to make rational decisions based on 
whether they are in a shared or common knowledge situation, 
two conditions have to be met: 1) the ability to accurately identify 
their knowledge situation 
and 2) the ability to deduce the strategic implications of these 
situations. 
While prior research suggests that 
people 
can reliable accomplish the 
second task, our findings imply they 
may struggle with the first. 
Future work 
should explore human decision making in more ecologically valid, real-life scenarios where shared and common knowledge are relevant, 
but 
participants are not explicitly informed about their knowledge levels.



If our experiments indeed 
reveal a curse of shared knowledge, as we conjecture, this is a somewhat surprising result. In most cases, agents' bounded rationality and lack of logical omniscience lead them to 
draw fewer conclusions 
than they 
otherwise 
might~\cite{fagin1995reasoning, halpern2011dealing}.
However, in our case, 
these limitations appear to have the opposite effect: They cause individuals to
conflate shared knowledge with common knowledge, 
leading them 
to conclude more than is warranted (e.g. believing 
common knowledge exists when 
it does not). 

From studies on coordination, particularly the coordinated attack problem~\cite{fagin1995reasoning}, we know that successful coordination generally requires common knowledge~\cite{moses2016relating}, 
which is often unattainable.
Common knowledge is relatively rare, 
typically arising only 
in situations of joint 
attention, such as during public announcements. 
Despite this, humans frequently attempt to coordinate actions and 
decisions based solely on
shared knowledge, 
accepting the inherent risk of miscoordination. In most real-life situations, the penalties for miscoordination 
are likely lower than in the Canteen Dilemma, 
making the risk more acceptable (though also in the coordinated attack problem, the penalty for miscoordination is severe). 
What the Canteen Dilemma highlights is that, in many cases, we may not even recognize the risk of miscoordination.
This suggests that humans might not consciously choose to 
rely on lower-order shared knowlege for coordination, but 
instead often
assume that coordination 
is
guaranteed. 
This assumption might stem from an inability
to distinguish lower-order shared knowledge from common knowledge. 
To our knowledge, this point has not been 
previously made. 

The `curse of shared knowledge' has potential implications for multiple disciplines that study human reasoning and decision making, including cognitive science (e.g.\ work on Theory of Mind), economics and game theory (e.g.\ behavioral economics and decision-making under uncertainty and information asymmetry), and artificial intelligence (e.g.\ in multiagent systems and human-robot interaction).

\clearpage

\appendix

\section*{\centering Appendix} 



\section{Experimental design and data collection} A large total sample size of $N=870$ was chosen to get robust conclusions from the statistical analysis while giving room for high variability in behavior. In fact, we thought that there would be a minority of pairs
opting for the all-office strategy. This turned out not to be the case. Experiments on Amazon Mechanical Turk had a total of 714 participants (including dropouts, see Appendix \ref{app:mturk-walkthrough}), while the two classroom experiments at the Technical University of Denmark (DTU1 and DTU2) had a total of 106 and 50 participants, respectively\footnote{Complete anonymized data files and all code can be downloaded from \url{https://github.com/gavstrik/Paper_canteen_dilemma}.}

The average payout to MTurk workers was \$4.17 (including a general participation fee of \$2). After accepting our task and providing informed consent, participants from MTurk were put in a ’waiting room’ until they were paired up with another participant. After an instructions page, detailing the rules of the game, participants were given an arrival time $t \in \{8{:}00, 8{:}10, 8{:}20, 8{:}30, 8{:}40, 8{:}50, 9{:}00,9{:}10\}$ and asked to make a decision between between going to the canteen or to the office. Next, participants were asked to estimate how certain they were that their `colleague' made the same choice as them, ranging from `very uncertain' over `slightly certain', `somewhat certain' and `quite certain' to `very certain', which were translated  into numerical values, $e_i$, used in the payoff calculations (see below). A results page was shown between each round, showing the results of the previous rounds, including arrival times for both players, their choices, their own certainty estimates and resulting payoffs. After 30 seconds, the game would automatically proceed to the next round. After the last round, we asked all participants a few final questions about their strategy and their understanding of the game. The experiments were implemented using oTree 2.1.35~\cite{ChenSchongerWickens16}.

The two classroom experiments, DTU1 and DTU2, differed from the MTurk experiment in a few aspects: 1) the maximum number of rounds played was increased from 10 to 30; 2) the initial bonus given each participant was increased accordingly from \$10 to \$30; 3) three additional questions were asked in order to elicit more explicitly some of the implicit assumptions and explicit behaviours by the students; 4) participants were told that they would not receive any monetary rewards, but that they should try to do their best. DTU1 received prizes. Screenshots, additional questions, experimental settings, and a detailed walk-through can be found in Appendix \ref{app:materials-methods-further}.

\subsection{Mixed effect logistic regression} \label{appendix:logistic} 
The experimental results were also analyzed using a generalized mixed-effects logistic regression model in R (\emph{glmer}), with the arrival time $t$ and experimental condition $exp$ (MTurk, DTU1, DTU2) as fixed effects including interactions and random effects for each pair of players, because we expect pairs to behave differently, depending on which levels of shared knowledge they individually take into account in their reasoning and collectively try to agree upon. Exploratory analysis including the round number as an additional fixed effect did not show any significant contributions, meaning that players did not significantly change their strategy during the game. 

\begin{figure} 
	\centering\includegraphics[width=0.8\linewidth]{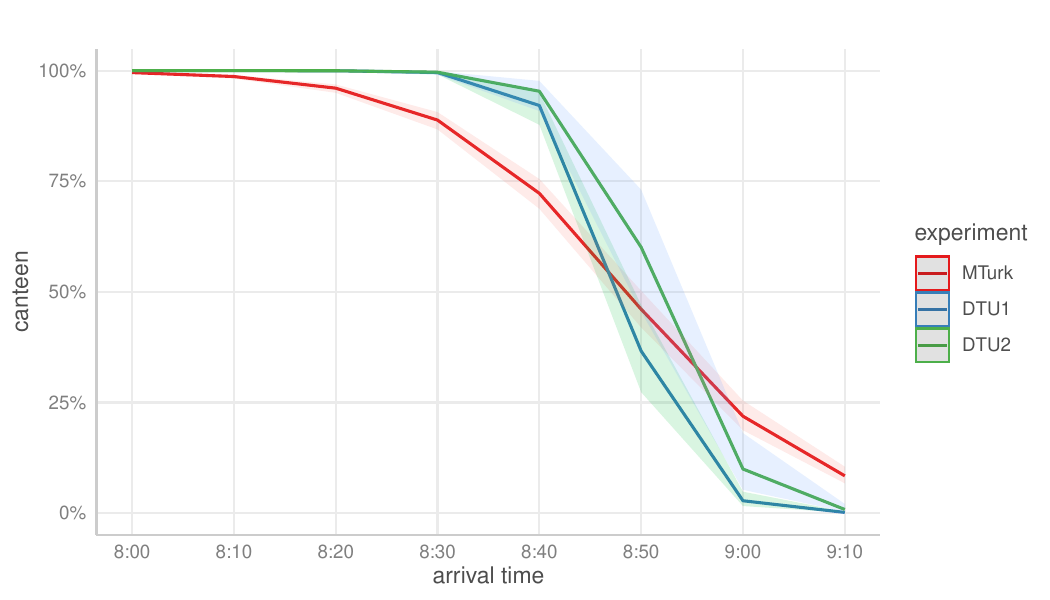}
	\caption{Percentage of canteen choices as a function of arrival times. The colored lines are logistic regression lines from a mixed-effects model with 95\% confidence intervals shown as translucent bands. Fitted parameters show significant differences in the slope and intercept between MTurk and DTU experiments ($p < .0001$).}
	\label{fig:mx}
\end{figure}

Using a generalized linear mixed-effects model with the arrival time and the three experimental conditions as fixed effects and with player-pairs as random effects, we obtain a logistic regression fit in Fig.~\ref{fig:mx}, showing the probability of players choosing the canteen as a function of their arrival time.Results from the mixed-effects  model are almost indistinguishable from the simple logit model without random effects used in Fig.~\ref{fig:1} The steepest slope occurs at $p(t) = 1/2 = -\alpha/\beta$, which for the MTurk and DTU1 experiments is $t=8{:}48$ and for DTU2 is $t=8{:}52$. The regression line in Fig.~\ref{fig:indistinguish} is obtained by combining observations from all three experiments and using a simple logistic regression model without random effects. The high number of observations imply small confidence bands. Hence, conclusions from the models can be viewed as robust. Code and data are available on github (\url{https://github.com/gavstrik/Paper_canteen_dilemma}).

\section{Payoffs and Penalties} \label{appendix:payoffs}
 All MTurk players finishing the game were paid a participation fee of \$2. In addition, a bonus could be earned if players did well. Before the game started, the bonus was set to \$10 for all participants. After each round, the bonus was reduced by a personal penalty, depending on the two players' choices. Penalties are calculated using a logarithmic scoring rule and by ordering them to be minimized by successful coordinations into the canteen. Penalties are maximized by any type of miscoordination or forbidden choice (i.e.\ going to the canteen at 9:00 or later). Office coordinations are designed to have larger penalties than canteen coordinations, but smaller penalties than miscoordinations in order to make sure that coordination remains the main objective of the game. Penalties are defined as negative utility values in the following way. First, we define the chosen action $a_i$ by player $i$, $i=1,2$, to take binary values encoding the canteen option ($a_i=0$) and the office option ($a_i=1$), and define their respective certainty estimates $e_i \in \{0.5, 0.625, 0.75, 0.875, 0.99\}$. We can then express the utility $u$ received by player $1$ as $u(e_1, a_1, a_2) = (1-|a_1-a_2| + a_1a_2)ln(e_1) + 2|a_1-a_2|ln(1-e_1)$, and symmetric for player 2. If any of the players choose the canteen at 9:00 or after, the utility becomes $u(e_1,a_1,a_2) = 2ln(1-e_1)$ for player 1 (and symmetric for player 2), corresponding to a miscoordination. As an example, imagine player 1 arrives at $8{:}40$ and choses the canteen, $a_1=0$. She estimates the probability that her colleague also will to go to the canteen to ``somewhat certain'', $e_1=0.75$. If her colleague indeed chooses the canteen, $a_2=0$, her utility will be $u(e_1,a_1,a_2) = ln(e_1) = -0.29$, but if her prediction proves false and her colleague chooses the office instead, her utility will be $u(e_1,a_1,a_2) = 2ln(1-e_1)=-2.77$. If she goes to the office just like her colleague, her utility is $u(e_1,a_1,a_2) = 2ln(e_1)=-0.58$. It should be noted that the logarithmic scoring rule used here is strictly proper. It is therefore not surprising that we find a good match between estimates and actual choices at arrival times different from those that are prone to miscoordinations, as seen in Fig.~\ref{fig:certain}, indicating that loss minimization remained a central concern and that participants made their choices and estimates as honestly as possible \cite{good1992rational, seidenfeld1985calibration, palfrey2009eliciting}.

\section{Formal Analysis} \label{appendix:formal}
We now give a detailed formal analysis of the game. We want to keep things as general as possible, so we don't assume a particular payoff structure, but any payoff structure satisfying certain constraints, and we want to consider arbitrary strategies, also mixed ones. This makes the analysis a bit non-trivial, although the underlying intuitions are rather straightforward. We of course haven't expected our human subjects to have made a deep analysis corresponding to the one below before starting to play the game. That should also not be necessary: The crucial aspect is to realize that there is no time point at which it is safe to go to the canteen (and hence avoid the high penalty), and this only requires to reason that 8{:}50 is not safe, and that then also 8{:}40 is not safe, and that then also... etc. What we do below is just to verify in a generalised setting that this reasoning is indeed formally correct and game theoretically valid and acting according to it leads to the highest expected payoff.  

The game can be represented as a game with three players, \emph{nature}, \emph{player 1} and \emph{player 2}. Nature is the player that initially decides the arrival times of player 1 and 2. Then player 1 and 2 are each informed of their own arrival time, and each have to choose among two actions: $o$ for going to the office and $c$ for going to the canteen. Based on the choice of actions by all three agents, player 1 and 2 receive a payoff. The payoff for player 1 is always equal to the payoff for player 2 (common payoff game). We  disregard the certainty estimates for now. The action choice of nature can be represented as an \emph{arrival pair} ${\bf t} = (t_1,t_2)$ consisting of the arrival time $t_1$ for player 1 and $t_2$ for player 2. Any arrival pair $\bf t$ has to satisfy that  $| t_1 - t_2 | = 0{:}10$. In our specific version of the game, we additionally have the restriction that $8{:}10 \leq t_i \leq 9{:}10$ for $i = 1,2$. The analysis of optimal strategies however doesn't depend on the exact arrival times available, so we will make things a bit more general and only assume that there is an earliest arrival time $\tmin$ and a latest arrival time $\tmax$, and that $\tmin \leq 8{:}50$ and $\tmax \geq 9{:}00$. Since we are counting in tens of minutes, by slight abuse of language we will call an arrival time \emph{even} if it is of the form $9{:}00+0{:}20x$ for some integer $x$, and \emph{odd} if it is of the form $9{:}00+0{:}10+0{:}20x$. The game then naturally splits into two disjoint subgames: in one of these subgames, player 1 always arrives at an odd time, and player 2 at an even time; 
 in the other subgame, the arrival time of the two players are reversed. The two subgames are disjoint in the sense that it is always common knowledge among the players which of the two arrived at an even time. We can hence restrict focus to one of the two subgames. We will pick the subgame where player 1 arrives at an odd time. To simplify further, we can without loss of generality assume that both $\tmin$ and $\tmax$ are also odd times. All reasoning performed in the following generalises immediately to the other subgame and to other values of $\tmin$ and $\tmax$. 
 
We let $T$ denote the set of arrival times and $\bf T$ the set of arrival pairs.  
The game starts by nature choosing an element $\mathbf{t} \in \bf T$. Nature is not a strategic player, so we assume that $\bf t$ is chosen uniformly at random, which is exactly how $\bf t$ is chosen in our experiments. The participants do not know that the arrival times are chosen uniformly at random, as this is left implicit in the description of the game. The following analysis of optimal strategies in the game could potentially change if arrival times were chosen according to a highly skewed probability distribution.

When nature has chosen its action ${\bf t} \in \bf T$ and player 1 and 2 have chosen their actions $a_1$ and $a_2$, player 1 and 2 receive their payoff, which we denote $u_\mathbf{t}(a_1,a_2)$ (the utility resulting from player 1 choosing $a_1$ and player 2 choosing $a_2$ given that nature played $\bf t$). We don't need to make any assumptions regarding the exact utility values (payoff values), except that successful coordination into the canteen before $9{:}00$, denoted $u_\mathit{c}$, is always better than successful coordination into the offices, denoted $u_\mathit{o}$, which again is always better than being miscoordinated, denoted $u_\mathit{m}$. Hence we suppose given fixed utility values $u_\mathit{c} > u_\mathit{o} > u_\mathit{m}$ such that for all ${\bf t} = (t_1,t_2) \in {\bf T}$,  $u_{\bf t}(o,o) = u_\mathit{o}$, $u_{\bf t}(o,c) = u_{\bf t}(c,o) = u_\mathit{m}$ and 
\[
u_{\bf t}(c,c) = 
\begin{cases}
   u_\mathit{c} &\text{if $t_1,t_2 < 9{:}00$} \\
   u_\mathit{m} &\text{otherwise}
\end{cases}
\]

A \emph{strategy} 
 is a mapping from arrival times to actions, that is, 
   a mapping $s: T \to \{c,o\}$. 
 Normally one would define two strategies, one for each player. That is however not necessary in our game. By assumption, player 1 can only observe odd arrival times, so if $t$ is odd, $s(t)$ denotes the strategic choice of player 1 at $t$, otherwise of player 2 at $t$. In this way $s$ encodes a \emph{strategy profile} (a strategy for each player). Since each player only observe their own arrival time, their strategic choice can only depend on their observed arrival time, which is why $s$ is a mapping from $T$ rather than from $\bf T$.   

Note that the defined strategies are memoryless (Markov strategies), that is, a player's choice only depends on the observed arrival time in the current round, not the history of arrival times and chosen actions in earlier rounds. Human players playing the game should not be expected to necessarily play memoryless strategies, as they might seek to adapt to the observed strategy of the other player. However, since it is a repeated game (every round is a new instance of the same game), perfectly rational players playing the game for a sufficient number of rounds should converge to an optimal memoryless strategy. We will leave further discussion of history-dependent strategies and focus on the optimal memoryless strategies in the following.


The payoff of a strategy $s$ at an arrival pair $(t_1,t_2)$ is then defined as $u_{(t_1,t_2)}(s) = u_{(t_1,t_2)} (s(t_1),s(t_2))$.  
The \emph{expected utility} $EU(s)$ of a strategy $s$ is the average of the payoffs~\cite{shoham2008multiagent}:
\[
 EU(s) = \frac{1}{|  {\bf T} |} \sum_{{\bf t} \in \bf T} u_\mathbf{t}(s). 
\]
Note again that player 1 and 2 get the same payoff (common-payoff game), so there is only one expected utility value to be computed. A strategy $s'$ \emph{Pareto dominates} another strategy $s$ if $EU(s') > EU(s)$ \cite{shoham2008multiagent}.
A strategy is \emph{Pareto optimal} if there does not exist another strategy dominating it. 
The game is cooperative (between player 1 and 2), so both players should seek to play a Pareto optimal strategy. Also, since it is a common-payoff game, all Pareto optimal strategies have the same expected utility~\cite{shoham2008multiagent}. 

Note also that our strategies are pure. To make things as general as possible, one would normally also consider mixed strategies, i.e., probability distributions over strategies (or, equivalently, mappings of arrival times into probability distributions over actions).  However, that is not necessary here as our game is a common payoff game. As explained in~\cite{leyton2022essentials}, common payoff games always have at least one pure Pareto optimal strategy. Furthermore, any Pareto optimal mixed strategy will be a combination of pure Pareto optimal strategies. 
To see the latter, suppose, for $i=1,2$, that player $i$ has a mixed strategy $m_i \in \Pi(S_i)$, where $S_i$ is the set of possible pure strategies of player $i$, and $\Pi(X)$ denotes the set of probability distributions over $X$~\cite{shoham2008multiagent} (in our particular game, $S_1$ would assign choices to only the odd arrival times, and $S_2$ to the even). Suppose further that the strategy profile $(m_1,m_2)$ is Pareto optimal. Let $S'_i \subseteq S_i$ denote the subset of strategies that are assigned positive probability in $m_i$. 
Each pure strategy profile $(s_1,s_2) \in S'_1 \times S'_2$ has an expected utility. Consider a pure strategy profile $(s_1,s_2) \in S'_1 \times S'_2$ that has maximal expected utility among all policy profiles in $S'_1 \times S'_2$. Suppose there exists another strategy profile $(s'_1,s'_2) \in S'_1 \times S'_2$ with strictly lower expected utility. If that's the case, then the pure strategy profile $(s_1,s_2)$ has higher expected utility than the mixed strategy profile $(m_1,m_2)$, which contradicts the assumption. Hence all strategy profiles $(s'_1,s'_2) \in S'_1 \times S'_2$ must have the same expected utility, hence all be pure Pareto optimal strategies. This means that we can focus on finding pure Pareto optimal strategies, since any potential Pareto optimal mixed strategies will be a combination of such Pareto optimal pure strategies.

We will now try to determine the possible candidates for being Pareto optimal strategies. 
\begin{lemma}\label{lemma:no-canteen-at-nine}
If a strategy $s$ is Pareto optimal then $s(t) =o$ for all $t \geq 9{:}00$.
\end{lemma}
\begin{proof}
We prove the contrapositive. Consider a strategy $s$ with $s(t) = c$ for some $t \geq 9{:}00$. We can assume $t$ to be odd, the other case being proved similarly. This implies $t \geq 9{:}10$.
Now define a strategy $s'$ which is identical to $s$ except $s'(t') = o$ for all $t' \geq 9{:}00$ (meaning that for arrival times at $9{:}00$ or after, both agents go to the office).  We want to show that $s'$ Pareto dominates $s$, implying that $s$ is not Pareto optimal. First note that
 $u_{(t,t+0{:}10)}(s') = u_o$
 by definition of $s'$.
For this arrival pair, it is impossible to receive the utility $u_c$, so $u_o$ is the highest possible payoff for this arrival pair. 
Hence any strategy not specifying $(o,o)$ at this arrival pair 
will have a strictly lower (expected) payoff. Now note that since $s(t) = c$, then $s$ does \emph{not} specify $(o,o)$ at $(t,t+0{:}10)$, and hence we get 
$u_{(t,t+0{:}10)}(s') > u_{(t,t+0{:}10)}(s)$. 
This proves the existence of an arrival pair for which $s'$ has a strictly higher utility than $s$. To prove $EU(s') > EU(s)$, we hence only need to prove that $u_\mathbf{t'}(s') \geq u_\mathbf{t'}(s)$ for all $\mathbf{t'} = (t'_1,t'_2) \in \bf T$. When $t'_1,t'_2 < 9{:}00$ this is trivial, as we then have $s'(t'_i) = s(t'_i)$ by definition of $s'$. When $t'_j \geq 9{:}00$ for some $j$, $s'(t'_j) = o$ by definition.
If $u_{\bf t'}(s) = u_m$, we immediately get $u_{\bf t'}(s') \geq u_{\bf t'}(s)$, since $u_m$ is the lowest payoff. If $u_{\bf t'}(s) = u_o$ then also $u_{\bf t'}(s') = u_o$, since $s'$ has an $o$ in all the places where $s$ has. 
\end{proof}

\begin{lemma}\label{lemma:no-c-after-o}
Let $s$ be a Pareto optimal strategy. For all arrival times $t$, if $s(t) = o$ then $s(t+0{:}10) = o$.
\end{lemma}
\begin{proof}
We prove the contrapositive.
Suppose $s$ is a strategy with $s(t) = o$ and $s(t+0{:}10) = c$ for some $t$. We need to prove that then $s$ is not Pareto optimal, i.e., we need to find a strategy $s'$ that  Pareto dominates $s$. 
If $s(t') = c$ for some $t' \geq 9{:}00$, the existence of a strategy dominating $s$ follows
immediately from Lemma~\ref{lemma:no-canteen-at-nine}. We can hence in the following assume that $s(t') = o$ for all $t' \geq 9{:}00$.  Since $s(t+0{:}10) = c$, we can thus also conclude that $t + 0{:}10 < 9{:}00$.
Now define $s'$ to be identical to $s$ except that we let $s'(t') = c$ for all $t' \leq t$.
We want to show that $EU(s') > EU(s)$. First consider the arrival pair $\mathbf{t} = (t,t+0{:}10)$. We immediately get $u_{\bf t}(s') = u_c > u_m = u_{\bf t}(s)$.
To prove $EU(s') > EU(s)$, it hence now only remains to prove that $u_{\bf t'}(s') \geq u_{\bf t'}(s)$ for all $\mathbf{t'} = (t'_1,t'_2) \in \bf T$. The only non-trivial case is when $t'_i \leq t$ for some $i$ 
(in all other cases, $s'(\bf t') = s(\bf t')$). 
If $t'_i \leq t$ then either we also have $t'_{3-i} \leq t$ or else $t'_i = t$ and $t'_{3-i} = t+ 0{:}10$. The latter case has already been covered, so we only need to consider the case of $t'_1,t'_2 \leq t$. Then $s'(t'_1) = s'(t'_2) = c$, by definition of $s'$, so $u_{\bf t'}(s') = u_c$. Since $u_c$ is the highest possible payoff, we immediately get $u_{\bf t'}(s') \geq u_{\bf t'} (s)$, and we're done.
\end{proof}

\begin{lemma}\label{lemma:canteen-propagate} 
Let $s$ be a Pareto optimal strategy. For all arrival times $t < 8{:}50$, if $s(t)= c$ then $s(t+0{:}10) = c$.
\end{lemma}
\begin{proof}
  We again prove the contrapositive. So suppose
$s$ is a strategy with $s(t) = c$ and $s(t+0{:}10) = o$ for some $t < 8{:}50$. We need to show that $s$ is not Pareto optimal. We can assume $t$ to be odd, the other case being symmetric. Consider the strategy $s'$ that is like $s$ except $s(t+0{:}10) = c$. Note that $u_{\bf t}(s') = u_{\bf t}(s)$ when $\bf t$ doesn't contain $t+0{:}10$. To prove that $s'$ dominates $s$, we hence only need to prove that $u_{(t,t+0{:}10)}(s') + u_{(t+0{:}20,t+0{:}10)}(s') > u_{(t,t+0{:}10)}(s) + u_{(t+0{:}20,t+0{:}10)}(s)$.
We have $u_{(t,t+0{:}10)}(s) = u_m$ and $u_{(t,t+0{:}10)}(s') = u_c$, so we need to prove $u_c + u_{(t+0{:}20,t+0{:}10)}(s') > u_m + u_{(t+0{:}20,t+0{:}10)}(s)$. Since $s(t+0{:}10) = o$, we get $u_{(t+0{:}20,t+0{:}10)}(s) < u_c$ (the highest payoff of coordinating into the canteen cannot be reached). Since the lowest payoff is $u_m$, we now get $u_c + u_{(t+0{:}20,t+0{:}10)}(s') \geq u_c + u_m > u_{(t+0{:}20,t+0{:}10)}(s) + u_m$, as required.
\end{proof}

\begin{definition}
A \emph{cut-off strategy with cut-off $t'$} is a strategy $s$ with $s(t) = c$ for all $t < t'$ and $s(t) = o$ for all $t \geq t'$. 
The cut-off strategy with cut-off $t_{min}$ is also called the \emph{all-office strategy}.
\end{definition}
Note that the strategy we in the informal discussions above referred to as the ``canteen-before-9'' strategy is the cut-off strategy with cut-off $9{:}00$. We now get the result on optimal strategies claimed in the informal discussion.
\begin{theorem}\label{theorem:all-office-or-cut-off}
Any Pareto optimal strategy is either the all-office strategy or the cut-off strategy with cut-off $9{:}00$.
\end{theorem}
\begin{proof}
Let $s$ be a Pareto optimal strategy. Then by Lemma~\ref{lemma:no-canteen-at-nine}, $s(t_{max}) = o$. Let $t_{cut}$ denote the earliest time with $s(t_{cut}) = o$. Then $s(t) =c$ for all $t_{min} \leq t < t_{cut}$. Suppose first that $t_{cut} > t_{min}$. Then $s(t_{min}) = c$. By Lemma~\ref{lemma:canteen-propagate}, we then get that $s(t) = c$ for all $t \leq 8{:}50$. By Lemma~\ref{lemma:no-canteen-at-nine}, we furthermore have that $s(t) = o$ for all $t \geq 9{:}00$. In other words, $s$ is the cut-off strategy with cut-off $9{:}00$.  Suppose alternatively that $t_{cut} = t_{min}$. Then $s(t_{min})=o$, and by  Lemma~\ref{lemma:no-c-after-o}, we then get $s(t) = o$ for all $t \geq t_{min}$. This is the all-office strategy. 
\end{proof}
The theorem proves what was argued in the main text: There are only two candidates for an optimal strategy, the all-office strategy or the canteen-before-9 strategy. This does not in any way imply that we should expect human players to adopt any of these two strategies, but if two perfectly rational players were to play the game, and if they knew they could expect the other player to play perfectly rational as well, and expect the other player to expect the same, and so on (i.e., we have common knowledge of perfect rationality~\cite{aumann1995backward}),
of course the optimal strategy would be played. And, as earlier mentioned, in our particular version of the game, the optimal strategy is the all-office strategy.




\section{Materials and Methods: Further Details}
\label{app:materials-methods-further}

\subsection{Mechanical Turk (MTurk) Experiment, Including Walkthrough and Screenshots}
\label{app:mturk-walkthrough} 


The MTurk experiments were approved by the the Research Ethics Committee at the Faculty of Humanities, University of
Copenhagen, Denmark on 22 February 2019. The approval stated: 
\begin{quote} 
The Research Ethics Committee at the Faculty of Humanities, University of
Copenhagen, has assessed your research experiment “The Canteen Dilemma“. Based
on the information you have provided in your Application for Ethical Approval, the
Committee has concurred that the project activities are in accordance with the relevant
International and Danish ethical guidelines and regulations. [...] The Committee hereby approves your project.
\end{quote} 
MTurk participants were recruited, and online experiments were conducted in the period from the 25th February 2019 to 1st March 2019. After accepting
our `Human Intelligence Task' (HIT) and providing informed consent, see screenshot in Fig.~\ref{fig:informed_consent}, 
\begin{figure} 
  \centering\includegraphics[width=0.7\linewidth]{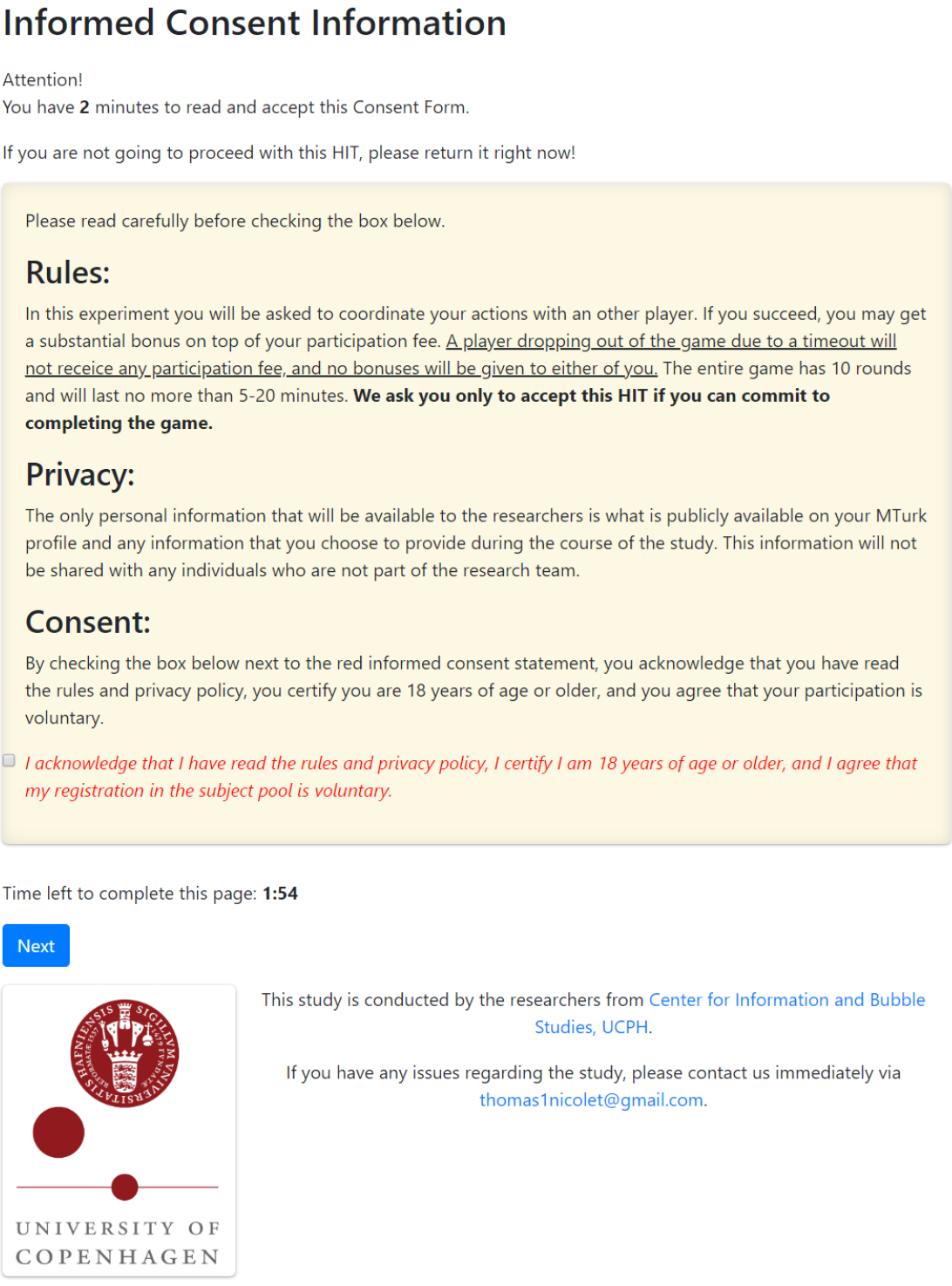}
\caption{Informed Consent Form.}
\label{fig:informed_consent}
\end{figure} 
participants from MTurk were put in a `waiting room' until they were paired up with another participant. After a group of two was formed, participants were directed to an initial instruction page which detailed the rules of the game with a time limit of 240 seconds, see screenshot in Fig.~\ref{fig:instructions}.
\begin{figure} 
  \centering\includegraphics[width=0.7\linewidth]{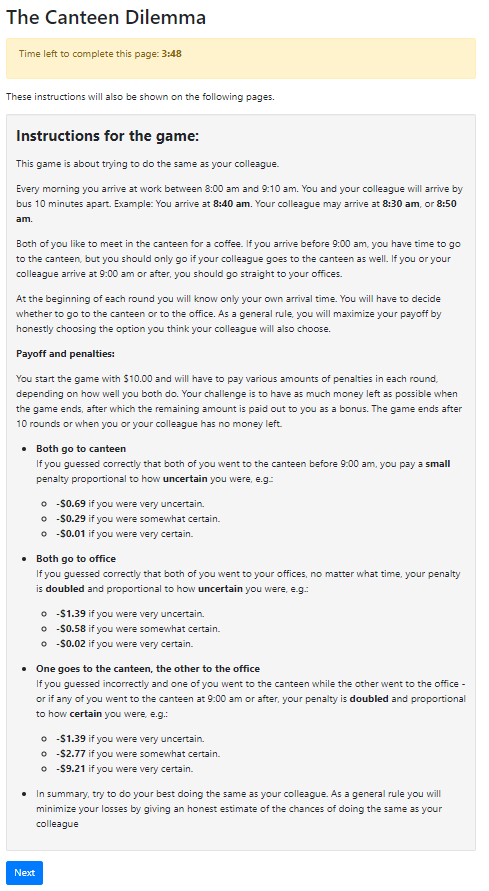}
  \caption{Screenshot of instructions page.}
  \label{fig:instructions} 
  \end{figure}
  \begin{figure}  
  \centering\includegraphics[width=0.7\linewidth]{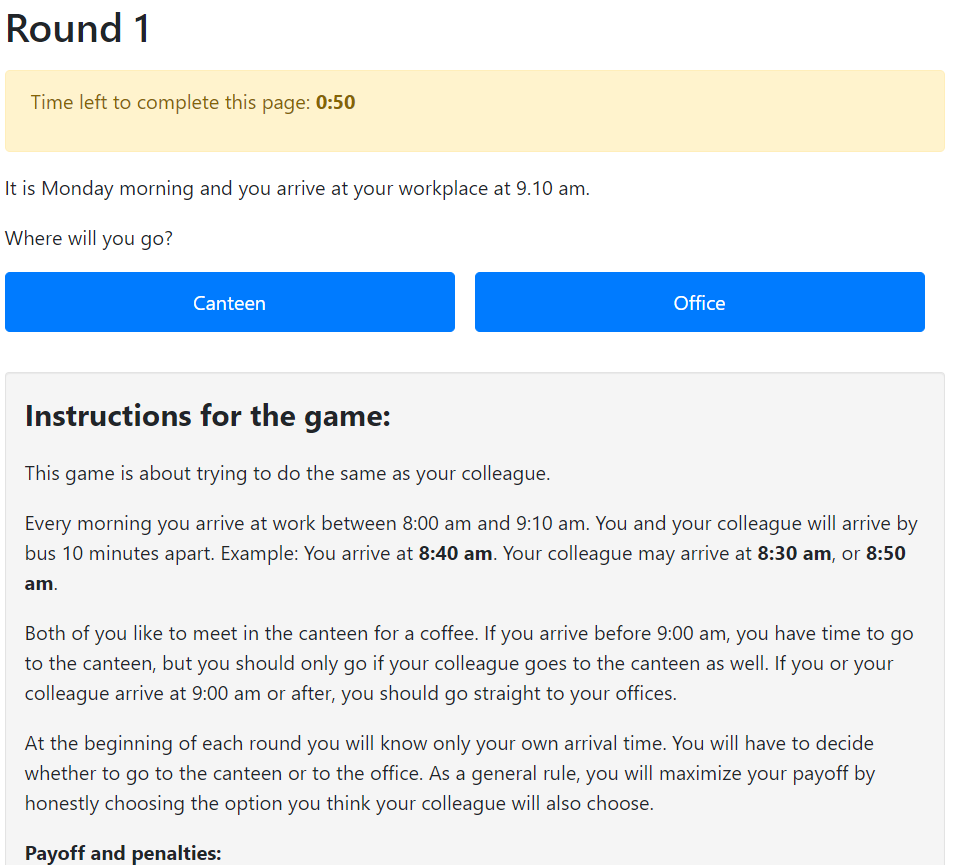}
\caption{Screenshot of choice page, round 1.}
\label{fig:round1}
\end{figure}
After reading the instructions, participants were directed to round 1 (of 10) where they were given their own arrival time and asked to make a decision between between going to the canteen or the office. Each round had a time limit of 61 seconds and rules from the instructions were repeated on the bottom of the page, see example screenshots for round 1 in Fig.~\ref{fig:round1}.
\begin{figure} 
  \centering\includegraphics[width=0.7\linewidth]{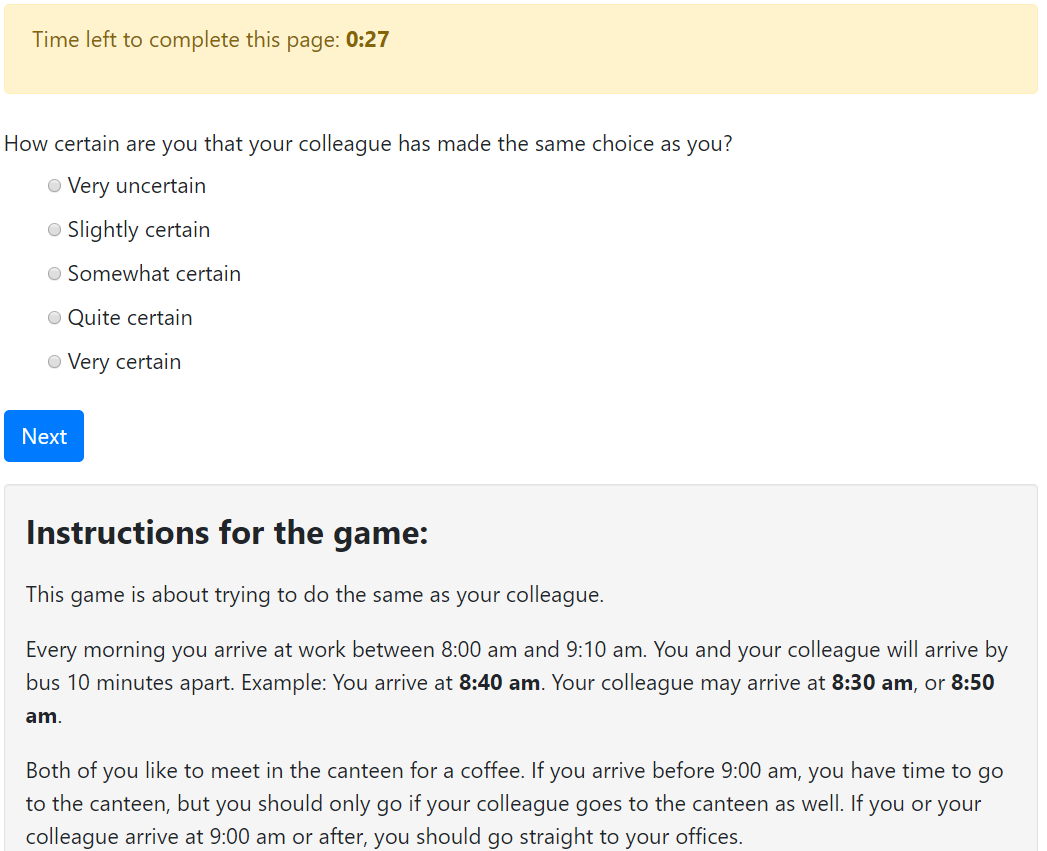}
  \caption{Screenshot of page where participants had to estimate the probability that their colleague would make the same choice.}
  \label{fig:certainty}
  \end{figure}
After making their decision (`Canteen' or `Office'), participants were asked to estimate how certain they were that the other player made the same choice as them, ranging from `very uncertain' over `slightly certain', `somewhat certain' and `quite certain' to `very certain', see Fig.~\ref{fig:certainty}.
\begin{figure} 
  \centering\includegraphics[width=0.7\linewidth]{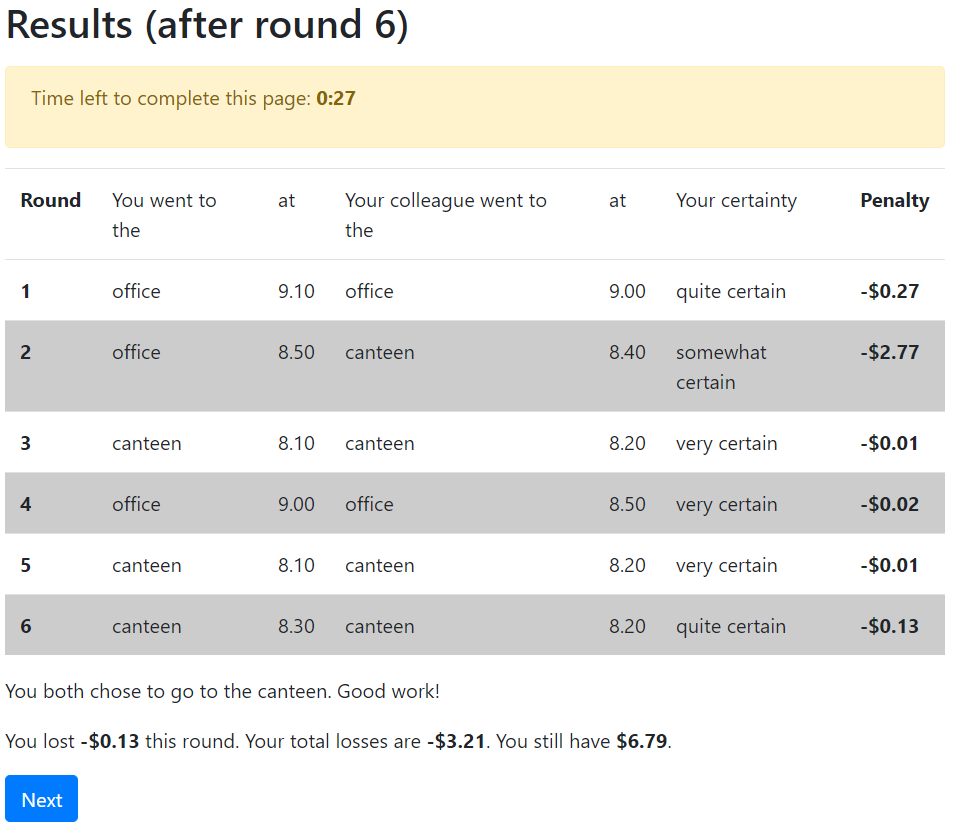}
  \caption{Screenshot of results page, round 6.}
  \label{fig:results6}
  \end{figure}
After both players have made their choices and their certainty estimates, they are prompted to a results page showing them the results of the previous rounds, including arrival times for both players, their choices, their own certainty estimate and resulting payoff, see example screenshot after round 6 in Fig.~\ref{fig:results6}. After 30 seconds, the game would automatically proceed to the next round. 
 
\begin{figure} 
  \centering\includegraphics[width=0.7\linewidth]{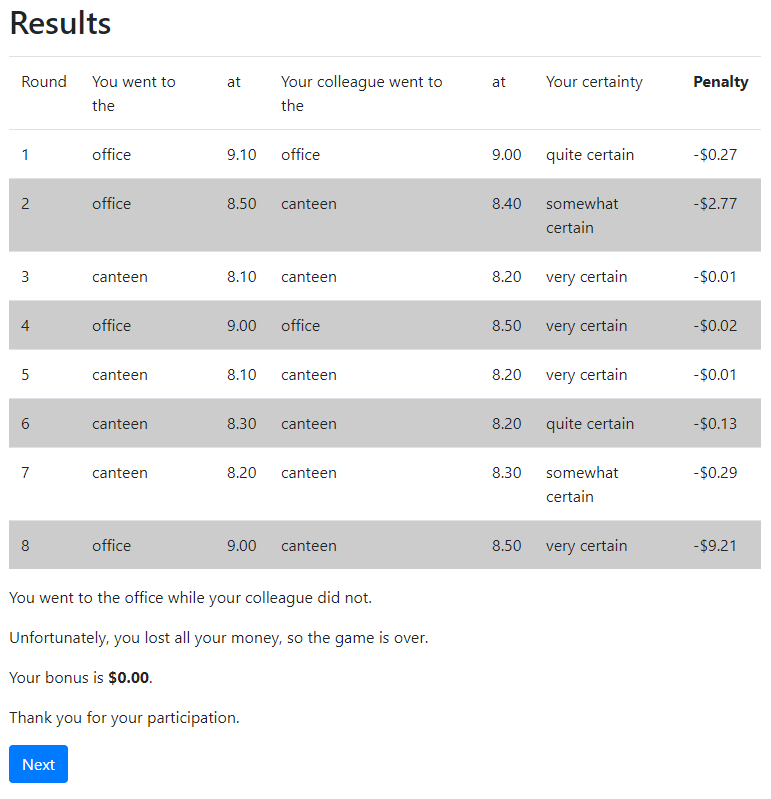}
  \caption{Screenshot after a player has lost all of her bonus.}
  \label{fig:round8finished}
  \end{figure}
In many instances, groups were not able to play the maximal number of rounds, because one or both of the participants had lost all their bonuses. An example of such a situation is shown in Fig.~\ref{fig:round8finished}. In such cases, the game would end for both players and they were asked to answer a follow-up question:

\begin{quote}
\indent
1. ``The game is over. Do you think it was your fault it is over, your colleagues fault, or do
you think it was because of some other reason?'' \textit{(Possible answers: `My fault', `Other's fault', `Other reason')}
\end{quote}
This question probes into participants' ability to rise above their possibly myopic understanding of the game. In addition to this question, all participants were asked three additional post-game questions about their strategy of play and their understanding of the game. The first was:
\begin{quote}
\indent
2. ``What strategy did you use while playing this game?" \textit{(open ended)}
\end{quote}
The answers to this question provided insight into the reasoning and thoughts of the participants. The next question was used to gauge the depths of recursive reasoning and reads:
\begin{quote}
\indent
3. ``Imagine you could have agreed beforehand with your colleague about a point in time where it is safe to go to the canteen. What time would that be?" \textit{(`I don't know', `There is no such time', 8:00, 8:10, 8:20, 8:30, 8:40, 8:50, 9:00, 9:10)}
\end{quote}
A final question pertaining all participant's understanding of the concept of common knowledge was the following:
\begin{quote}
\indent
4. ``Imagine you arrive at 8:10 am. Is it common knowledge between you and your colleague that it is safe to go to the canteen, that is, you both arrived before 9:00 am?''. \textit{(`Yes', `No', `Don't know')}
\end{quote}

\subsection{DTU Experiments}\label{app:dtu-experiments}
The DTU students were asked three additional post-game questions:
\begin{quote}
5. ``Did you ever go to the canteen at an arrival time later than what was safe according
to your previous answer? Why or why not?'' \textit{(open ended)}
\end{quote}
\begin{quote}
6. `Did you ever choose differently after seeing the same arrival time again at a later point
in the game? Why or why not?" \textit{(open ended)}
\end{quote}
\begin{quote}
7. `Imagine you arrived at [8:40/9:00] and you have been secretly informed that your colleague’s arrival time is 8:50. Where do you think your colleague will go?'' \textit{(`Canteen', `Office')}
\end{quote}
In the last question, half of the participants were given 8:40 as their own arrival time while the other half were given 9:00. The question concerns whether player’s own knowledge of the other’s arrival time affect their prediction of the other player’s decision. It relates to the curse of knowledge \cite{birch2007curse} since participants might attribute their own belief (that it is early enough or too late to go to the canteen) to the other player.

\subsection{MTurk Settings}
\label{MTurksettings}
Looking at Table \ref{table:A1}, the average payout to MTurk-workers was \$4.17 (including a general participation fee of \$2) which amounts to an average of more than \$20 per hour, which is considered very generous according to MTurk guidelines and certainly above the recently estimated average of \$6 per hour when excluding un-submitted and rejected work \cite{HaraEtAl18}. Students in the DTU experiments (DTU1 and DTU2) did not receive any monetary reward, but were told to try to maximize their payoff, and awarded prizes for doing well.

\begin{table} 
\centering
\begin{tabular}{lccccc}
\hline
Experiment   &   Participants &   Attrition rate &   N & Rounds & AvgPayout  (\$) \\
\hline
 MTurk 	& 714 &  0.02 &  680 & 10 & 4.36    \\
 DTU1  	& 106 &  0.13 &  80 & 30 & (prizes)   \\
 DTU2  	&   50 &  0.08 &  42 & 30 & - \\
\hline
\end{tabular}
\caption{The main experiment on Amazon Mechanical Turk (MTurk) had 714 participants of which 17 participants (2.4\%) quit prematurely, some of them quite early in the game. These quitters were told (in the consent form) that they would receive no bonus and no participantion fee. They are excluded from the data analysis. Their ``lucky'' colleagues however, got both their bonuses and participation fee, but are likewise excluded from the data analysis. Therefore the final number of subjects, $N$, is reduced by twice the attrition rate. In the two DTU experiments with students from the Technical University of Denmark (DTU1 and DTU2) attrition rates were slightly higher, mainly due to the higher number of rounds played.}
\label{table:A1}
\end{table}

Participants quitting a study before completing it is prevalent on MTurk, and varies systemically across experimental conditions \cite{ZhouFishbach16}. In our experiments on Amazon Mechanical Turk attrition rates were $2\%$, witnessing that we had managed to design the experiment in a way that minimized drop-out rates. A combination of high payouts, a logarithmic scoring rule taking advantage of loss aversion biases, a consent form stipulating the revocation of the participation fee after dropout, and minimization of waiting times may have been the main reasons.

All participants automatically received a `qualification' when accepting a HIT. This qualification ensured that participants could not play the game again. In addition, we required that participants should have had completed at least 500 HITs, have an accepted HIT rate of 98\% or above, and should be from the United States or Canada. This ensured that we would get relatively experienced and qualified participants.

MTurk participant attention was expected to be equal to or better than undergraduate participant's attention \cite{Rand2012}, while various forms of dishonesty (practical joking or trying to pair up with a friend) was expected to be rare, due to the high turnover rate experienced for our HITs. In addition, during the experiment, participants had easy access to our email for questions and possible bug reports. Apart from a few timeouts, participants had no comments or complaints.

\subsection{Full Rule Set for DIS Version of the Game}\label{app:dis-rules-full} 

This game is a 2-player cooperative coordination game with 15 rounds. The 2 players have to try to coordinate their actions, and if they are successful, they will both achieve the same positive reward (so it's not a zero-sum game). The positive reward can either be \$1 or \$2. If you're miscoordinated, you lose all the money you earned so far in the game. After the 15 rounds, you get to keep the rewards you have accumulated. Well, actually, in this version of the game, you will not be paid actual money, but you should try to behave as if you did, i.e., try to maximize the reward.

In each round, each player is dealt a face down card. On the face of the card there is a number in the interval 2--10 (it's a card from a deck of standard playing cards). The numbers on the two dealt cards are always exactly 1 apart, so if for instance one of the players gets a 3, then the other will either be getting a 2 or a 4. And since 10 is the maximum, if one player gets a 10, necessarily the other player is getting a 9. And similarly, if one gets a 2, necessarily the other gets a 3.  

When the cards are dealt, each player looks at the number on their own card without showing it to the other. The two players are not allowed to communicate or exchange any kind of information during the game. After each player has inspected their own card, they should hide either a white or a black marble stone in their hand and put the hand on the table. Each player has to make an individual choice without communicating with or revealing it to the other player. 

When both players have their closed hands with the hidden stones on the table, they wait for the facilitator to tell them that now they can open their hands to reveal the two stones, and they can also turn their cards to reveal their numbers.  

Both players then receive rewards depending on the colors of their stones and the numbers on their cards. The white marble stone is worth \$2, and the black marble stone is worth \$1, but you only get to keep your stones (money) if you choose the same stone as the opponent, and additionally, if you both choose the white stones, you only get to keep them if both card numbers are strictly below 9. In more details:  
\begin{enumerate}
\item[a)] If one player chose a white stone and the other a black stone, then they both lose all the stones (money) they have received in the game so far.
\item[b)] If both players chose a black stone, they both receive \$1 (i.e., a black stone).
\item[c)] If both players chose a white stone, then they both receive \$2 (i.e., a white stone), but only if both numbers are strictly below 9, meaning they are both in the interval 2--8. If not, the players lose all the money (stones) they have received in the game so far.
\end{enumerate} 

\section{Results} 

\subsection{Transcriptions of Qualitative Interviews with DIS Students}

\paragraph{Interview 1}
This is a transcription of an interview between Robin Engelhardt (R) and students (S) from the DIS experiment.

\medskip
R: .. Okay. I'll ask you some questions, and then maybe if I don't pin anybody of you out, just the first one who would like to answer answers. Okay, so the first question would be 
“Did you believe that your co players received the same rules as you yourself?” 

S: “Yes.“

R: “Right. That should be quite obvious. And that was the problem with the games online, because some people think, oh, yeah, but we can invent new rules for each player. They cannot know that the others know. Okay, next question:

R: “Were you uncertain about the rules and whether your co-payer had understood the rules?” 

S: “Yes.”

R: “Okay. How may we start with you? In what sense were you uncertain?” 

S: “I don't know. It took me reading it over a few times, I guess.” 

R: “Okay, and then you thought your coplayer might as well have some problems understanding it?”

S: “Yeah, kind of about being on the same page.“

R: “And when did you think you were okay, you both knew how the rules were? At what point? In which round?”

S: “One round. Because we both, like, played it and then realized this is why we did it.” 

R: “What about you? You said also yes.” No?”

R: “Okay. So, what about you? Did you think I didn't understand the rules? So did you think you were completely clear about the rules?” 

S: “I think we both understood it.” 

R: “Okay.” 

S: “You know, I wasn't uncertain.”

R: “You weren’t uncertain at all? Okay. That's the benefits of being in the same room, I guess. So, do you think it was your fault, the game, you miscoordinated? Or was it your partner's fault? Or do you think it was because of some other reason? So now anybody of you we miscoordinated once, we all miscoordinated at one point. Right. Whose fault was it? Whose error?”

S: “I don’t think it was anybody’s fault. It was the strategy?” 

S: “I don't think she ever made a decision that I would think was objectively a wrong decision.” 

R: “And you never thought that. You thought that about me at one point when I chose the black stone instead of the white one. You had a seven, I had a six. I could have had an eight. Okay, so now the more substantial. Yeah, but I would like to hear …”

S: “because we are trying to remember what happened, we're trying to figure out the scenario. One sort of scenario is neither of our faults. “

S: “We kind of end up developing the strategy that because we had so much so late in the game, might as well play it safe. “

R: “So, what was safe?”

S: “Black.“

S: “Yeah, black is safe, always. Right. Any combo. If both people have black.” 

R: “Yes. You didn't do that. Okay. Would you have done that? “

S: “You'd have to be able to communicate with your partner. Which technically we weren't supposed to do. “

R: “Yes, exactly.”

S: “Otherwise, you'd lose out on a lot before your partner figured out you were trying to do it. “

S: “If I was playing with somebody who I knew better and had played games with before and knew if they played more riskier, more safe, I think it would have been easier to know what we're doing without being okay. “

R: “When did you find out that the black was a safe strategy? “

S: “I think we always did. But I think the big thing, as we started to get more and more altogether, we didn't want to lose them all. So, then I think that's when it hit me that I was like, oh, I just play black.”

S: “In round eight, or something like that?”

S: “Yeah. Okay. I know we're in the same room as a bunch of other people playing it, but I would see that we had more than everyone else and I kind of wanted to stay on top of everyone else. “
S: “Okay, so it's due to the comparison of the others you found out maybe safe way is to always choose black even though it is a lower payoff.”

R: “Did you other guys think about that?” 

S: I thought that black is dangerous because if one person plays black, you lose everything because the other one probably chooses white.”

R: “So, what was your strategy then? 

S: “Also, to play most conservatively, but there's like near like the 7-8-9-7-8 range that's like then I usually play up for black. Below that I play for white.” 

R: “Okay, so let me understand. If you have an eight, you take black?”

S: “..and I also think like two spots in either direction compensate for what the partner would think. “

R: “Okay, so seven, you would choose black as well? “

S: “Yeah, I think so.”

R: “Six, you would choose white? At least that you did, I remember. “

S: “Seven you would have to think about they might have an eight and then…

S: “ But if they have a six…Right. If you play the seven, assuming that yeah,..”

S: “…it was kind of weird because I was thinking because with eight, statistically speaking, it makes more sense to pick black because if they have a nine, they're always going to pick black. If they have a seven, they might anticipate you, they might not. Technically, there's a chance that they choose black still. So if you look at it from peer statistics, it makes more sense to pick black eight, but by the same token, as a seven, it makes more sense to pick white because the 6 is basically always going to pick white, and on the 8th they might pick white, they might not.”

R: “Yeah…. Did you have an idea of what a safe strategy would be? What strategy did you use?”

S: “I don't know. I've decided, as I saw the plan, lower than a certain number, I would only use white. If it was like a higher number, I would think about it.”

R: “What was the cut-off number?”

S: “Seven.”

R: “So seven would be white?”

S: “I don't know.“

R: “Okay, I'll find out. Okay. So, it seems that most of you chose some kind of probability measure about what to choose when having a seven or eight. Is that true? Or did you think, okay, black is always safe, so why not black, for instance? “

S: “I don't think the money or is enough for me to want to care. The black was a million dollars and the white was like 2 million, like double. I would always choose black because I'm happy with a million dollars, but if it's one and $2, I don't care about one and $2. “

R: “Okay…”

S: “ so I'd rather have more than like ..:”

R: “But if you would have chosen black all the time from the beginning, if it were a million dollars, you might have a guy choosing white and you'll just bust. “

S: “Yeah, I know, but I'm saying I feel like if the money was more, I want to play, like, a safer strategy, but since the money is negligible, I don't care about selling the safe strategy. “ 

R: “Yeah, but my question is what is safe? What is the most safe strategy? Right. Even though you might know that black, if both could agree on always choosing black, it would be fine, but it's not sure. There still wouldn't be the most safe strategy, right? Because you don't know. Maybe the guy chooses white once in a while when he gets a two or three.”

S: “I feel like maybe if it's like 15 rounds, if I played like the first all ten rounds, all black, then he would get the hit and at least you'll get 5 million in the end and I'd be pretty happy walking away.”

S: “There's also the cut off ones where if you had a two or three, you would obviously play.  There'd be no sense in playing black. There's not even a chance that they have one that white wouldn't work out. If you had a two, then you know the other person is a three or three. You know they have two or four. At that point you would play white because you would do more because that's safe no matter what.”

S: “The safest is always play white, like two to five there's no risk involved in the numbers…. “ 

R: “There's absolutely no risk involved. … from five low or from where?”

S: “From five and lower.”

S: “Well, that's where I was thinking about that a couple of times because it gets a little weird in that. Like…”

R: “Do you personally have five and lower?”

S: “Except since you're saying there's risk involved in six, there's technically also risk involved in five. “

S: “Okay, if they have a seven, they might think you have an eight. That's where the risk comes. It doesn't go lower than that. If you have a six, you could be mistaken for having an eight, and that's where the risk comes from. If you have a five, you're not going to be mistaken for having any …”

S: “Yeah, but the point, too, is you could decide with that six that since you might be mistaken for having an eight, that you should go black.”

S: “Two to five is white, six to ten is black, except and then you always win.”

S: “No, because if you have a six and a five and you use that paradigm, you lose. “

S: “No, but it's if you personally have a five if I personally have a five, then I know that they either have a six or five. So even if I know the max, they can have six. Even when they're a six, they're still going…”

S: “ to but if the six then goes, okay, it makes sense for me to go black because they could think that I could have eight. “

S: “Then it doesn't work”

S: “No, the cut off is weird. I was thinking about that earlier. I didn't go super in depth, but I was like, yeah, if you dig into this deep enough, it could like, cascade and be like, well, at this cut off, they could think ..

S: “Ok, this it's not 100\% like they …”

R: “You used cut off at seven, right? We did both more or less right, though I failed once. Do you think? Now we are talking about cut off. It's very interesting because you were, at least in the start, sure that it would be completely safe with five below to take white. If you could have talked beforehand before the game started, what is our cut off? And you could have decided it is five, would you then always win?” 

S: “No, it's always possible.” 

S: “Having one number be 50 50. “

R: “Yeah, maybe that's a good idea. What about if you had a three? Do you think it is common knowledge that both of you are below eight?”

S: “Yes”

R: “..and common knowledge in the sense of common knowledge, in the sense of logic or in the sense only of common sense? “

S: “I think it's logic, because if you have a three, the max they have is four. So, the max they could think you have five, which is six. “

R: “So, you're studying logic. What is it called? AI. What is it? AI. So, you have a lot of logic with Thomas, so you understand the concept of coming out. You have heard about it. Everybody knows that. Everybody knows broadcasting messages. So, we expect everybody to know that. “

R: “Did you did you ever choose differently after seeing the same number again at a later point in time? So, imagine you chose seven. You chose white, getting a seven, which was white, and then later you get a seven and you chose black, something like that. So, getting the same number, choosing differently later, did we do that? “

S: “Yeah, I did that with you. After we lost, all of ours were low again. We were playing it safer and wanting to get more money when it was safe”.

R: “So, you played it safe. More black choices when you had a lot of money?”

S: “Yeah, because we didn't want to lose everything we got.”
 
R: “You did that?”

S: “The more money we got, the things that we got, and then when we lost it, we reset to just like ..”

R: “Okay, imagine you get a number, say, seven, and you have been secretly informed that your colleague's number is eight. What do you think?

S: “That happened to me…. she was thinking, and I had a seven, so I'm like, if she had a six, she wouldn't have been thinking. Yes…” 

R: “And so you read her, cold reading?”

S: “Hm..mm”

R: “Okay. And you said she has an eight, and she is contemplating black.“

S: “Yes.”

R: “And what happened?”

S: “I played a white. She played a black. Like, you would actually know what decision she will make. “

R: “Yes. Did you learn that choosing white with the number eight was dangerous? Yes? I guess everybody likes that. If yes, why was choosing white with the number seven not dangerous?”

S: “It was.”

R: “It was less dangerous?”

S: “Yeah, it's less dangerous because with eight, there is a 50\% chance that they're guaranteed to choose black with the seven there's a 50\% chance that they're almost guaranteed to choose white, because most of the time…white.”

S: “But logically, if they choose black with eight, it's also not very safe to choose white with seven. Right? “

S: “Yup. Seven, if you chose white, you know that your partner should choose white and could choose white. But the only thing you were going against, like, your partner. “

S: “Yeah. Your partner might think you have nine. Right.?”

S: ”….”

S: “Okay. And it's even less dangerous if you have a six?”

S: “Yes”

R: “Human psychology… Logically speaking, doesn't make sense, but because we just think in probabilities….”

S: “I think we were all playing simpler rather than overthinking. “

R: “That's what we do.” 
“
S: "Your partner would have to overthink so hard and take so long to do it to get to the point where it'd be worth it for you to put a black. “

R: “Maybe that's a good reason. So we don't expect others to overthink it, because we don't do it ourselves. So even though we might somehow have a hunch that having six means that the other guy has seven, meaning that he might think I have eight, meaning he might think that I think he has a nine,..”

S: “You can work your way all the way down.”

R: “Yeah. And that's why only black is safe in the end. Okay, thanks a lot. I think we're done. “

\bigskip
\paragraph{Interview 2}
This is a transcription of an interview between Thomas Bolander (T) and students S1, S2 and S4 in the DIS experiment.

\medskip
T: Did you get surprised sometimes about the choice of the co-player? 

S1: We had one moment, I think we where I just had not quite considered all the ramifications of… I think I had a six and you had a seven. 

S2: Yes. So my thinking was that possibly that they had had an eight, so they would think that I had a nine, so therefore I should play a black stone.

T: Yeah. 

S1: Whereas I was not thinking about what you would think I would think. And so I was like, oh, well, you'll either have a seven or a five. So that's what I should… and then I did not consider the psychology there.

T: But after that?

S1: After that, yeah. Then, like, we were much more, like, one of us has, like, a six or a seven. We were playing black stones very consistently. 

T: Right, but then would you still sometimes play a white stone? 

S1: I don't think we did.

S2: We never lost money after that. 

S1: Yeah, I think in that scenario, we played black stones each time.

S2: Yeah, so there were situations where, I remember one time, one of us had a seven, the other had an eight, and obviously we easily could have won \$4 instead of two. However, we both decided to play black stones in that situation because I think we both valued that, right? It would be better if we got two less dollars than if we lost all of our money. 

T: Yes, I agree. Yeah. If you want to maximize your reward. That seems sensible. But did you then… So was the reasoning kind of… ehm… 

S1: Just in any scenario where the other person could, like, even think that you had one of them think that you had a nine? 

S2: Yeah, or an eight, then we just played black. 

T: Okay, but then after you got this, would there be a situation where you would still have played white? Like, if you had a two, for instance? 

S1: Oh, yeah. I mean, we would still play white if it was like a two or three or four and a five.

S2: Up to six. 

S1: Yeah, up to six. 

T: Up to six? Okay. And you both thought in the same way?

S1: Yeah. 

S2: Yeah.

T: Okay. But is that then a good strategy to play to play white up to six? 

S1: I think so. If you're just trying to maximize, like, you're not going to get perfect score, but you're also not risking losing at all at that point, like losing all of your money.

T: But are you not risking it? 

S1: What do you mean? 

T: Okay, so this is kind of a strategy where you said up to this point, right, I'm playing white, and beyond that, I'm playing black. So the question is if it’s a safe strategy where you are not risking to lose all your money, or is it still a risky strategy? 

S1: It's a safe strategy, but you're not maximizing your reward.

T: Is it wise if it's six or below? Is that it?

S2: I think generally that’s how we were applying. 

T: But then if one of you gets a six and the other one gets a seven. 

S1: Well, if that person had a seven…

S2: I suppose that it just ended up working.

S1: Yeah, we didn't end up with that scenario, actually. 

T: No, that was your luck, then. 

S1: Actually, that did happen. We both played black. 

S2: Yeah. 

T: Okay. Yeah. 

S2: Where we have both had black? Because I think one of us thought that the other had.

S1: You know, I had the six. Oh, yeah. Because it was the time after I had initially fucked up and so then I was on a high alert of, like, oh, shit, don't want to do that again. We had won, like, one or two more rounds and then had that scenario, I think. 

T: Yeah. So you somehow changed your strategy a bit during the game?

S1: It was not that consistent. I was also trying to, like, I had realized, like, three or four rounds in also that I was just continually holding the black stone in my left hand and the white stone in my right hand and then realized, fuck, I can't do that, and then started shuffling the stones before I put one on the table.

T: Yeah, okay, okay. And what kind of strategy did you play? 

S3: I feel like it was a little weird because we got a lot of, like, four and five. We got very few that were sort of, like, on the boundary. And I think once we did, it kind of broke down because we got like, seven in a row that were like the highest card was like a six for like, seven rounds in a row. So we had all these two white stones built up and then it was like an eight and a seven, and we miscommunicated. 

T: And what is the boundary? 

S4: I think at the very first point where we miscommunicated because I had played, I believe, a white stone because I had a seven. And I was saying to myself, like, well, if I have an eight, then it's possible he has a nine. So until it's an eight, I'll play white. But I didn't think of the fact that he had an eight and would be thinking I could have a nine. And so then after that point, I had to start thinking more about whether you had. like, eight or sevens and then thought, I could have an eight and play black. And then I think that's where in the middle, we kind of started trying to figure out what our strategy was. 

S3: My thought process was, like, when you get an eight, it's like there's a 50\% chance that they have a nine and we'll guarantee black, and then there's 50\% chance they have a seven, but within that, there's a chance they'll play black and a chance they'll play white. So black felt like the safe option on an eight, just like because it's got better than 50\% odds was sort of my approach on that. 

T: Yeah, right, so can you say that again? So if you got an eight, is that it. 

S3: Yes. So if, like, I have an eight in my hand…

T: …and then they get a nine, and then you have to play and then they would necessarily play black. So in that case, you know, you would have to play black. But then you say if they get a seven, it's 50/50.

S3: Yeah. They might play black or white. So it's like them playing black is more likely when I have an eight was my thought process.

T: Yeah. That I see. But then if if you think they get seven, how do you deduce that then the the chance of them playing white is 50 and black 50? 

S3: I don't necessarily know 50/50. I just know there is some amount of chance they'll play black on a seven… +50\% of the time on a nine.

T: But did you also think in terms of probabilities? 

S1: Yeah.

S2: Yeah, I guess… I suppose I did. I mean, I think it was more to some extent. 

S4: I didn't think of probabilities. 

T: No. So what did you do? 

S4: I think I was the whole time just trying to think of a rule, at which point to stop playing white to start playing black.

T: And did you find the rule? 

S4: There isn't really one because then it would just keep going backwards and just play black all the time. To be 100\% safe. 

T: Yeah. But did you then consider to play black all the time? 

S4: Not once the actual game had started. I think you mentioned right before the game that you were also thinking maybe we could just play black the whole time – and it would have ended up better for us. But you feel like if you have safe cards, then you might as well play white. 

S1: It's like a sunk cost fallacy of, like… 

T: indeed it is a fallacy. Right. Because if you compute the expected reward, right? Then the problem is that the penalty of miscoordination is so high that it doesn't pay off. But was that clear, in a way? Let me also ask that: Was it clear from the rules that if you were, like, miscoordinated just once, then it would probably not pay off and it would be better to play just black/black all the time? 

S1: The way I was thinking about it was that, for me, I was considering halfway through the game, switching to just playing black, because once you make it past the halfway point in my mind is where it's worth considering that…

T: Okay, to have something left at the end, at least. 

S1: Yeah. Because if you do that early, like, if you fuck up early in the game… I should not curse for your interview.

T: No, that's fine. I will put it in the paper and I will have your name and address…. haha…. 

S1: I completely lost my train of thought. Oh, yeah. If you mess up early in the game, like, it's more recoverable because you have more than half of the game left. But once you have… If you have earned already half the game, then switching over I think is potentially beneficial. Like, if we could have communicated in advance, honestly, it's probably, like… or if we were to play another round and we're able to communicate in advance, I'd be like, let's figure out the vibes. And then if we get to the halfway point on this round, then we only play black for the rest.

S2: But also it's necessarily what needs to happen. Because if we were supposed Thomas gave us a four and a five, we obviously know that it's within the boundary. So at that point we could play. I think that it’s mainly…

S1: … It’s for the person that has a five. They could think that the other person has a six. And if they think the other person has a six, then they could think that that person has a seven and so on. 

S2: Nesting through all the four loops, if we go through that, the person has a five thinks the other person has a six, which means that the other person has a seven. Seven is within the boundary. 

S1: True. 

T: But then you can iterate once more, right? 

S2: Yeah. But then that would be null for the fact that the cards are a four and a five. 

T: Okay. So one question related to this. So imagine that you could before you started playing the game, you could have agreed with your co-player about a number, and then you say, okay, if it's below this, then you choose white, otherwise black. What would that number be? 

S1: Five. 

S2: No, I think it has to be six.

S1: But then you have the six and seven situation of when the person has seven, it's a little dicey.

S2: No, because if the person has seven. That means the other card is either six or eight. 

S2: Yeah, but if they're concerned that the other person has an eight, then they're also concerned that that person thinks they have a nine and will play black. Because whenever I had an eight, I would play black. 

S2: Yeah.

S1: And therefore if you're playing white there, it would be a mismatch. Right. 

S3: I don't think setting like a specific target would work. Frankly, just because it'll always kind of cascade out at a certain point. Like, oh, they could have this. And so I feel like more of like a loose rule, kind of like if you're confident, we'll lean on white and then after the halfway point, just like stick with black, I think would maybe be better. If you have a two, three, four, if it feels right, more than just, like, here's the hard limit. 

T: But if it's about what feels right, right? Then that might… Either what feels right is that always certain numbers feel right and certain doesn’t, right? And it's static. It's true all the way through the game or it changes over time. And if it changes over time, then I guess you're worse off than if it's fixed, right? Because you make yourself unpredictable to your co-player, right? And therefore it seems to me that you would have to say, okay, but then what feels right? It has to be somehow a fixed number, right? And then the question again is what number that is?

S3: But yeah, I mean, maybe like six. It's got to have a little bit of padding from the actual from nine, I feel like. 

T: Yeah, but the problem again is… The thing is, because, even if you agree on it, you say, okay, anything six or below that's black, right? Then there will be a situation where one has six and the other one has seven, and then you're miscoordinating. So in reality, and I think you also said it, right, at some point you said that it's cascading all the way. And I think you also said that you also talked about, like, playing all black, right? That you somehow, because…. That actually it is the only safe strategy, right? You cannot choose a cut-off point, because no matter what the cut off-point is, then one can be on one side of the cut-off point and the other on the other… so… Right, so, another question related to this. So imagine you get a three. Is it then is it then common knowledge between you and your co-player that it's safe to choose white? 

S1: Yeah. 

S2: Yes. 

S3: Yeah. 

T: Okay. Why? 

S3: Having a three or a four in your hand, I feel like those are both very confident. 

S1: There were two.

S3: Yeah, well, I was just sticking to the optimal… if you have a three and the other person has a four, that's where the question comes in. And I feel like even in that case, like a four is like miles away and so it's pretty, like, it's not going to cascade that far, I feel like.  

S1: What's the name of the psychology thing where they put people in a room and they're, like, there's a pile of money in this room, if all of you take it, like none of you win. Thing…

S3: Sounds like the marshmallow experiment. 

T: No, but it's different. No, maybe it's not. No, the marshmallow, then you get more marshmallows. If you don't eat it, you get more when you are done. But this is, like you said, about… all it depends on is whether everybody takes money. Is that it? 

S1: Yeah. So there's a pile of money in the room and if everyone wants to take a chunk of it, nobody gets the money. Whereas everybody… If all but one, or any other mixed people take money, as long as someone declines to take money, then other people get to have it. This game felt a bit like that in a weird way in my head. 

T: Okay. Yeah, no, but I guess all these coordination games is about trying to predict what the other player does, right? 

S1: Oh, it was the Milgram experiments… is what was going on in my head this entire time.

T: Yeah. Okay, so you talked about, like, it would be safe. It's around six or seven or something. It depends a bit. But you don't necessarily agree on that, then?

S4: I feel like you can't really set a limit. But there are cards like twos, threes and fours where it's just unspoken that you play white because there's no way that the other person will think you have a really high card.

T: Okay, right, okay, so you would say, like, two, is… that's safe… white? 

S4: Yeah. 

T: But then three is also safe white. And then…

S4: I know it's really contradictory, but every time we had fours and fives, we would always play white. We would never consider that the other person would think that many steps ahead. Or at least I didn't think that many steps ahead. 

T:  Okay, so you're somehow relying on… you think that there's a limit to how many iterations the other player will do. Is that it? That you somehow think that even if you understand that in principle, I can think about you thinking about me thinking about you thinking about me… I can continue…. then you would somehow naturally say, yeah, but that is unrealistic. Is that what you say? You wouldn't believe other people to do so.

S2: Most people aren't really thinking even one step ahead from their partner, let alone, like, two. 

S1: Yeah, I think I was capping it too unconsciously. I was like, the second you said that, I was like, yeah, I was assuming that you were thinking two steps ahead.

S2: Yeah, most people are probably thinking one to two steps ahead. I highly doubt a whole lot of people are thinking three. Three is very uncommon. I think anyone thinking more like four, five, six steps ahead is just kind of obsolete at that point. 

S1: They need to start playing instead. 

T: Okay, but let's say you were playing 100 rounds. That would be a bit boring, but let's say. Now after we discussed, what would you do? And you get the money. So it's money you get, and you can go out and celebrate afterwards. 

S3: So it's real money? All black. 

S2: Me too.

S1: Yeah. 

S3: Well, unless it's like two or three, all black.

T: But two and three is white? But the problem is then, if you get a three and you play white. So even if you agreed on this strategy, your co-player could still get a four.

S2: But for four, you can still play white.

S3: Yeah, I guess, like, if two iterations ahead of thinking is safe, then you're good. 

S2: You're stopping at, like, two. So suppose the person had five, the other person would either have to have six or a four. So if they had a six, that person would think, all right, either five or seven. So at that point, they know who you're with as well. 

T: Thanks a lot. Perfect. Thank you.

\bigskip
\paragraph{Interview 3}
Interview at DIS with Thomas Schrum Nicolet and four studentsS1, S2, S3 and S4 (S1 and S2 played together, as did S3 and S4)

\medskip
T: Did you believe your coplayers played with the same rules as you?

S1: Yeah.

S3: Yes.

S4: Yes.

T: Kind of like a standard assumption from the beginning?

S1: After the first loss yeah.

S3: Yeah.

S4: Yes.

T: By the first loss, do you mean like a strategy?

S1: The first time we put down different colored stones.

T: We are asking about the overall rules, did you play with the same rules?

S1: Oh yeah, yeah.

T: Were you uncertain about the rules, what you had to do, or uncertain about if your coplayer understood the rules?

S1: No.

S3: No.

T: You think it was pretty clear?

S1: Yeah.

S3: Yes.

T: Do you think it was your fault the game was over quickly, your partners fault, or something else?

S4: Mutual.

S3: It was mutual.

S1: Yeah, mutual.

T: What strategy did you have when you played this game, if you had a strategy?

S1: I played white for anything below 7, not a very risk tolerant strategy.

S3: Yeah, I think we started kind of aggresive with our whites, but once we lost first, we started being a little bit more conservative.

S3: I think if it is a 7 or lower I would put a white, and 8 was the only one where I was unsure.

T: So you [S3] did kind of 7 or lower is white, and for you [S1 + S4], 6 or lower would be white, so 7 or above would be black.

T: Did you have any kind of mixed strategy, like a probabilistic strategy? You [S3] 8 was up in the air, so that might be random?

S3: Yeah.

S4: Yeah.

T: And you [T1] had a maybe more determined strategy?

S1: We were determined, but I think if it had been closer between us and the next group, and we wanted to win, then probably...

S3: For the first three, we got 7-8, 7-8 and then 8-9. Even though it was random, we got too many 7-8s, so we're gonna see what's 8-9.

T: So when you're playing this game did you try to find the optimal choice or making probabilistic choices. I think we already like talked about this but did you try to like find some optimal solution, or something, like you say, you would have changed perhaps?

S1: I don't think I would have changed.

S2: No, I don't think I would have changed.

S3: No...

T: So imagine you could have agreed beforehand with your co player about a number where to do white and where to do black? What number would that be and does such a number exist?

S1: I would say 7.

S3: I would do 6.

S1: 6 or 7 yeah.

S4: I dont really think it exists per se because I think you could just move it anywhere and then if it's an eight it's between seven and nine but if you moved it down to seven you still don't know whether they have a six or an eight. I think it's hard to pick up. Because you still don't know which side of it it is on. I tried to think about that earlier on but I just assumed that you would just keep moving the line in either direction and that wouldn't really be helpful.

T: Okay

T: So imagine you get the number three, is it then common knowledge between you and your coplayer that it is safe to choose white?

S1: Yeah.

T: That you both have a number below 9?

S3: Yeah.

T: Did you ever choose white at a card number that was large than what was safe according to what you said earlier? You said you would choose 7?

S2: Do we?

S1: I think we chose seven once, because we got 7-8 twice before, so I was thinking we would get 7-6.

T: So you think that 7 wasn't safe, so you should choose a black stone, but you actually chose a white stone?

S1: Yeah.

T: Do you [S3 + S4] have any experience like that?

S3: Whenever I got like a 7 I did black.

T: So you didn't deviate from that?

S3: Not really, I think that was why we lost a couple of times.

T: Did you ever choose differently after seeing the same number again at a later point?

S1: Yeah we got 7-8 twice, so then I assumed we had 7-6

T: It's also hard to remember.

S3: I think on the third, like 8-9 or 7-8 one. We got 7-8 on the first two and then got 8-9 on the third one, so on the third one we both went for the black.

T: So you kind of miscoordinated in the beginning and then you changed...

T: I did your numbers [S1 + S2], so I remember you had 8-7 and then you miscoordinated in the beginning and then you had it later and I think you had both black then.

S1: Mh.

T: Did you learn that choosing white was dangerous when you get a number 8 ?

S3: Yes.

S4: Yes.

S1: Yeah I think that's how we lost.

T: Yeah, because your coplayer might get 9 and play black. Ok so, if you think that, why was choosing [black at] 7 not dangerous? 

... If you think 7 is also dangerous, so you might have picked black still at 7?

S1: Oh because if your partner has 8, then he could think you have 9.

T: Yeah, but then the question is, if you think that 7 is dangerous, why was choosing white with a 6 not dangerous?

S1: Ehm...

S3: Because 5 and 7...

S1: Yeah your partner can only have 5 and 7, and if he has 7, you can only have 6 or 8, so both are in range.

T: Yeah, but if you choose black at 7... and you think 6 is safe, so you choose white...

S1: Mhh.

S2: Mh.

S4: That's what I meant before why you can't pick one number that is safe, because you can just keep moving it down.

T: Yeah. So that was all the questions, unless you have anything else?

S1: Did you find an optimal number?

T: [Explains that no such number exists.]

\section{Additional supplementary material}

\subsection{Formal definitions of Private, Shared and Common Knowledge}
\label{A:definitions}
We can define the notion of common knowledge and related notions a bit more precisely as follows, following the conventions from epistemic logic (see e.g.\ Herzig and Mauffre~\cite{herzig2015share}). Given a proposition $p$ and an agent $i$, we use $K_i p$ to denote that agent $i$ knows $p$. Given a group of agents $G = \{1,\dots,m\}$, we say that $p$ is \emph{private knowledge} in $G$ if at least one of the agents know $p$, that is, if $K_i p$ is true for some $i \in G$. We use $E_G p$ to denote that everybody in $G$ knows $p$, that is, for all $i \in G$, it is true that $K_i p$. Whenever it is not necessary to be explicit about the group of agents $G$, we will just write $E p$ and say ``everybody knows $p$''. For all $n$, we then recursively define $E^n p$ to be shorthand for $E E^{n-1} p$, where $E^1 p$ is shorthand for $E p$. So for instance $E^2 p$ expresses that ``everybody knows that everybody knows that $p$'', and in general $E^n p$ means we have $n$ iterations of ``everybody knows that'' in front of $p$. We read $E^n p$ as ``everybody knows $p$ to depth/order $n$''. We also call this \emph{shared knowledge} (of $p$) to depth/order $n$, or \emph{$n$th-order shared knowledge}. When we say that $p$ is \emph{shared knowledge}, we mean that it is shared knowledge to depth $n$ for some $n \geq 1$. \emph{Common knowledge} of $p$ then means that $E^n p$ is true for all $n \in \mathbb{N}$.

In epistemic logic, the three notions---private, shared and common knowledge---are usually not considered to be mutually exclusive. So if $p$ is common knowledge, it is also automatically both shared and private, since when the conditions for $p$ being common knowledge are satisfied, also the conditions for it being shared and private are satisfied. However, in many cases, as in our paper, we want to make an exclusive distinction between the three types of knowledge. We can define $p$ to be \emph{shared knowledge only} if it is shared knowledge but not common knowledge. Thus, $p$ is shared knowledge only if for some $n$ we have $E^n p$ but not $E^{n+1} p$. Similarly, we can say that $p$ is \emph{private knowledge only} if $p$ is private but not shared knowledge. Thus, $p$ is private knowledge only if $K_i p$ holds for some, but not all, $i$. In most texts, as in ours, it is left implicit whether private and shared knowledge are interpreted inclusive or exclusive, that is, one doesn't explicitly distinguish between ``shared knowledge'' and ``shared knowledge only''. Normally it is clear from the context whether one intends the concept to be interpreted exclusively or inclusively. 

In our paper, we interpret the concepts exclusively, although we make an exception for private knowledge. When $p$ is known by all agents, we say that $p$ is both private and (first-order) shared knowledge. The exact border between private and shared knowledge vary significantly between different papers. De Freitas et al.~\cite{de2019common} consider the case $E p$ to still only be private knowledge, and for $p$ to be considered shared knowledge furthermore requires that there is at least one agent $i$ knowing $E p$ to be true (that is, requires $K_i E p$ to be true for some $i \in G$). The point of De Freitas et al.\ is that if only $E p$ is true, it is not really shared knowledge, but only private knowledge held by everyone in $G$. In our paper, we have sought a compromise between the terminology by De Freitas et al. and the standard terminology in epistemic logic, and hence we have the overlap between private and shared knowledge.

\subsection{Supplementary data analysis}
\label{app:supplementary-data-analysis} 
\begin{figure} 
\subfloat[\label{fig: logit mturk}]{%
  \includegraphics[width=.46\linewidth]{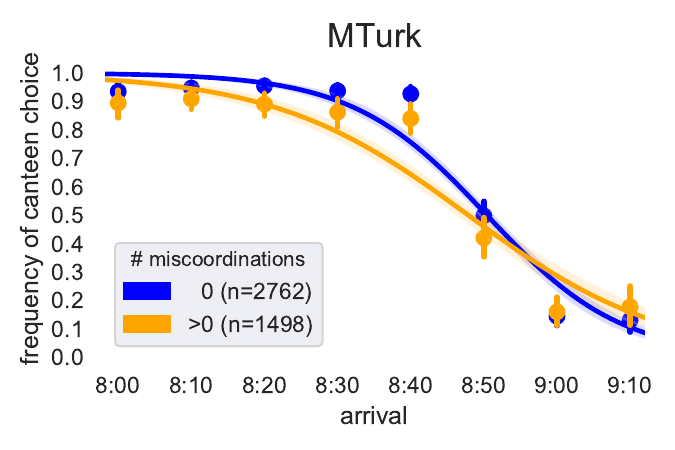}%
}\hfill
\subfloat[\label{fig: logit dtu}]{%
  \includegraphics[width=.46\linewidth]{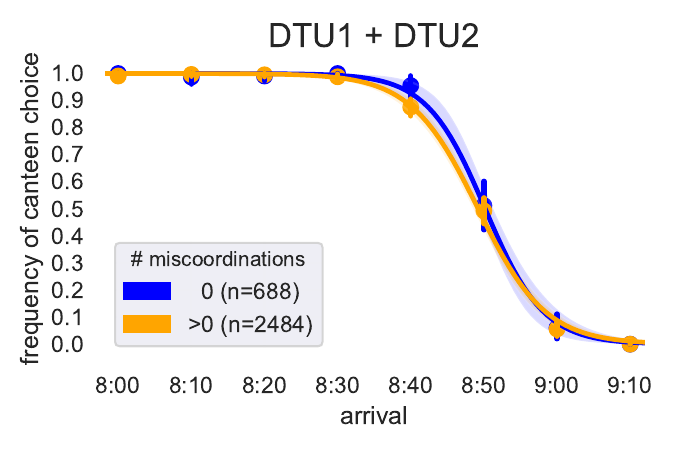}%
}\hfill
\caption{Participant's decisions of going to the canteen as a function of their arrival time, here partitioned into those groups who previous have experienced zero (blue) or one or more (orange) miscoordinations. MTurk participants are shown one the left, DTU students on the right}\label{fig:miscoordinations-appendix}
\subfloat[\label{fig: certain mturk}]{%
  \includegraphics[width=.46\linewidth]{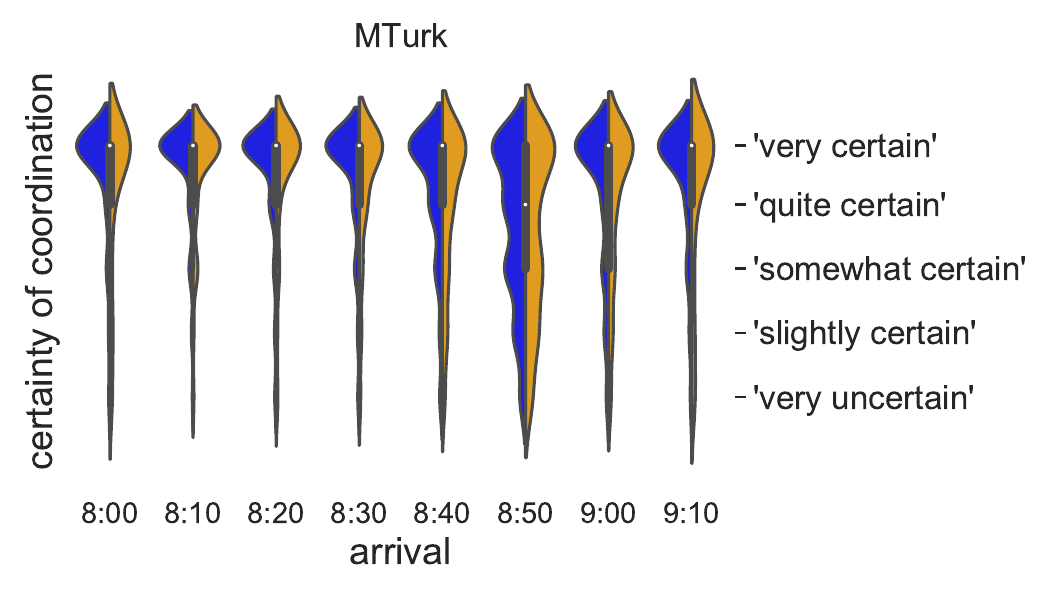}%
}
\subfloat[\label{fig: certain dtu}]{%
  \includegraphics[width=.46\linewidth]{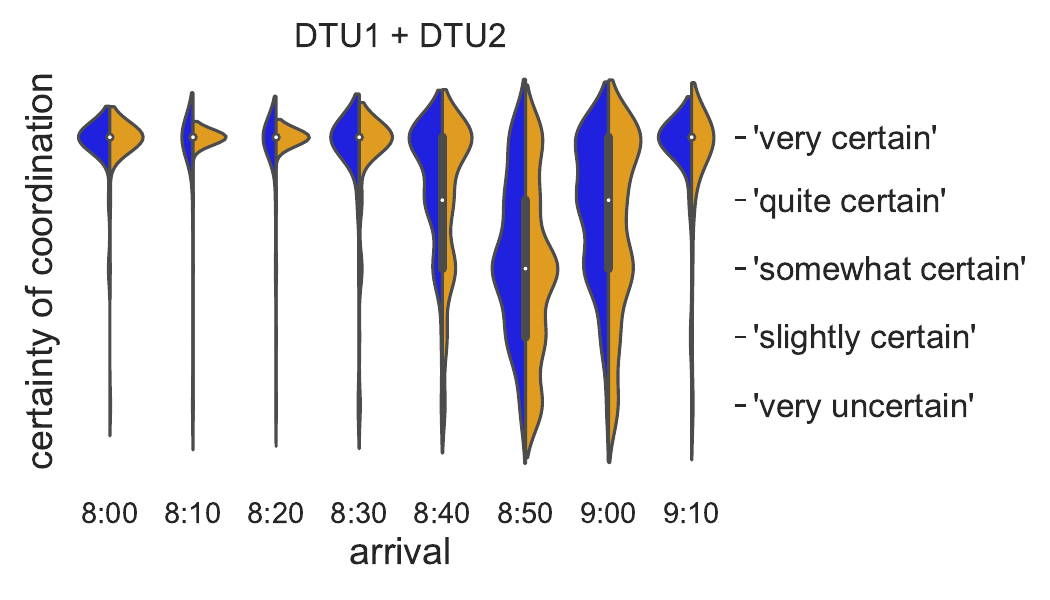}%
}
\caption{Participant's certainty estimates of being able to coordinate with their colleague as a function of their arrival time, partitioned into those groups who previous have experienced zero (blue) or one or more (orange) miscoordinations. MTurk participants are shown one the left, DTU students on the right}\label{fig:certainties}
\end{figure} 
When a group experiences a round of miscoordination, we expect some kind of learning to take place. `Why did my colleague choose differently than I did?' should be an obvious question a player asks herself, prompting deeper perspective-taking and possibly an understanding of the lack of common knowledge. We investigate this by partitioning decisions into those in which a participant never before has experienced a miscoodination with her colleague ($m=0$) and those in which a participant has experienced one or more miscoodinations ($m>0$). Furthermore, since the number of rounds - and hence miscoordinations - are very different for MTurk participants and DTU students, results from MTurk and DTU are shown separately, with the two DTU experiments combined. Results for the choices in Fig.~\ref{fig:miscoordinations-appendix} show significant differences between MTurk-groups having miscoordinated or not, while for DTU students those differences are not significant. Looking into the corresponding certainty estimates, as shown in Fig.~\ref{fig:certainties}, we see a similar pattern as before: Much steeper profiles among DTU students, but also much lower certainty estimates around critical arrival times. 

\section*{Author Contributions and Acknowledgments}

Thomas Bolander developed the game, made the theoretical analysis, and contributed as a lead author; Robin Engelhardt and Thomas S.\ Nicolet designed the experiment, analyzed the data, and contributed as lead authors. The authors declare no conflict of interest.

Server infrastructure and devops was handled by Mikkel Birkegaard Andersen. The authors wish to thank Vincent F. Hendricks for enabling the project. This research was approved by the Institutional Review Board at the University of Copenhagen and included informed consent by all participants in the study. The authors gratefully acknowledge the support provided by The Carlsberg Foundation under grant number CF 15-0212.


\bibliographystyle{plain}
\bibliography{../cd}

\end{document}